\documentclass[12pt]{article}

\usepackage{graphicx}
\usepackage{amsmath}
\usepackage{amssymb}
\usepackage{amsfonts}
\usepackage[doublespacing]{setspace}
\usepackage[ignoreall]{geometry}
\usepackage{scalefnt}
\usepackage{nicematrix}
\usepackage{tabularray}
\usepackage{verbatim}
\usepackage{appendix}

\newtheorem{theorem}{Theorem}

\newtheorem{proposition}{Proposition}
\newtheorem{remark}{Remark}

\newif\ifblind

\newenvironment{proof}[1][Proof]
{\noindent\textbf{#1.} }
{\ \rule{0.5em}{0.5em}}

\makeatletter
\renewcommand{\@makefntext}[1]{%
\noindent\makebox[1.8em][r]{\@makefnmark}\footnotesize #1}
\makeatother

\title{\Large Locally Robust Kernel Specification Tests for Conditional Moment Restrictions}

\ifblind
  \author{}
\else
  \author{
  Juan Carlos Escanciano\thanks{Department of Economics, Universidad Carlos III de Madrid, email: jescanci@eco.uc3m.es.}
  \\[0.2cm]
  \textit{Universidad Carlos III de Madrid}
  }
\fi

\date{\today}

\begin{document}

\maketitle

\vspace{-0.5cm}

\begin{abstract}
We develop kernel-based specification tests for semiparametric conditional moment models with high-dimensional nuisance parameters, extending existing conditional moment tests---which typically require asymptotically linear nuisance estimators---to accommodate modern machine-learning methods. The proposed locally robust kernel tests combine Neyman-orthogonal moments, cross-fitting, and reproducing kernel Hilbert space methods, yielding inference that is first-order insensitive to nuisance estimation error. We establish oracle equivalence between the feasible and infeasible test processes under local alternatives and weak nuisance-rate conditions, and characterize the resulting local power. A fast multiplier bootstrap avoids nuisance re-estimation. Applications include specification testing in high-dimensional linear and logistic regression, significance testing with machine-learning regressions, and tests of constant conditional treatment effects. Monte Carlo simulations and an application to the National Supported Work program illustrate the finite-sample performance of the proposed tests.

\bigskip

\noindent
\textit{Key words and phrases:}
Model checks; High-Dimensional models; Kernel methods;
Orthogonal moments; Multiplier bootstrap.
\end{abstract}

\clearpage

\section{Introduction}

Many semiparametric models are characterized by conditional moment restrictions involving high- or infinite-dimensional nuisance parameters. Examples include high-dimensional linear and generalized linear models, significance testing, and tests of homogeneous conditional treatment effects. In modern applications, these nuisance parameters are routinely estimated using machine-learning methods such as Lasso, random forests, boosting, and neural networks. Although these methods offer considerable flexibility, they create difficulties for specification testing: they may converge slowly, need not admit asymptotically linear representations, and can introduce first-order estimation error into the limiting distribution of conventional plug-in tests. Consequently, such tests need not control size when nuisance functions are estimated by machine learning.

This paper develops specification tests for semiparametric conditional moment models that remain valid with generic machine-learning nuisance estimators. We construct Neyman-orthogonal empirical processes indexed by a reproducing kernel Hilbert space (RKHS) and combine them with cross-fitting. The resulting locally robust kernel (LRK) tests are first-order insensitive to nuisance estimation error and require relatively mild product-rate conditions for their validity.

Our main theoretical result establishes process-level oracle equivalence: uniformly over the RKHS indexing class and under local alternatives, the feasible cross-fitted orthogonal process is asymptotically equivalent to its infeasible oracle counterpart. This result reduces the asymptotic analysis of the feasible test to that of the oracle process and extends orthogonal machine-learning methods from finite-dimensional parameter inference to global specification testing of conditional moment restrictions. We also characterize the local power of the tests and identify the alternative directions that remain detectable after orthogonalization. Finally, we establish the validity of a computationally efficient multiplier bootstrap that avoids repeated nuisance estimation.

The paper is related to three strands of the literature. The first studies specification testing based on integrated conditional moments and empirical processes; see, for example, Bierens (1982), Stute (1997), Chen and Fan (1999), Delgado and Gonz\'alez-Manteiga (2001), Escanciano (2006), Song (2010), He et al. (2026), and Tan et al. (2026), as well as the local-smoothing literature initiated by H\"ardle and Mammen (1993) and Zheng (1996). These methods typically require asymptotically linear or otherwise sufficiently regular nuisance estimators. In contrast, our framework accommodates generic machine-learning estimators through orthogonalization and cross-fitting. He et al. (2026), in closely related work, assess machine-learning goodness of fit using out-of-sample residual dependence measures. Our method instead uses orthogonal moments and remains valid under standard cross-fitting without requiring the testing sample to be asymptotically negligible relative to the training sample.

The second related literature concerns Neyman-orthogonal inference and testing, including Bickel et al. (2006), Escanciano and Goh (2014), Chernozhukov et al. (2015, 2018, 2022), Shah and B\"uhlmann (2018), Jankov\'a et al. (2020), and Bravo et al. (2020). Bickel et al. (2006) is particularly close in spirit, developing process-level semiparametric $C(\alpha)$-type tests for likelihood-based models. We instead treat conditional moment restrictions with generic machine-learning nuisance estimators and establish oracle equivalence for RKHS-indexed orthogonal processes.

Third, the paper contributes to RKHS-based testing; see, e.g., Gretton et al. (2012), Muandet et al. (2017), Chwialkowski et al. (2016), Sancetta (2022), Escanciano (2024), and Escanciano and de U\~na-\'Alvarez (2025). Existing orthogonal RKHS tests for conditional moment restrictions generally consider finite-dimensional nuisance parameters or finitely many moments. We allow high- and infinite-dimensional nuisance parameters and infinitely many moments while retaining the computational simplicity of kernel methods.

We illustrate the framework in three settings: specification testing in high-dimensional linear and logistic regression, significance testing with generic machine-learning regressions, and testing homogeneity of conditional average treatment effects. Monte Carlo experiments show satisfactory size and power across these settings. An application to the National Supported Work program finds no statistically significant evidence of treatment-effect heterogeneity across observed pre-treatment characteristics.

Section~2 presents the framework and motivating examples. Section~3 constructs the orthogonal empirical process, Section~4 introduces the LRK statistic and multiplier bootstrap, and Section~5 develops the asymptotic theory. Section~6 presents the simulations and empirical application, and Section~7 concludes. Proofs and additional results are provided in the supplementary material.

\section{Framework and Motivating Examples}

This section introduces the conditional moment framework, defines the nuisance parameters as population functionals of the data-generating distribution, and presents three motivating examples.

\subsection{Conditional Moment Framework}

We consider semiparametric models defined by
\begin{equation}
E\!\left[\varepsilon(W_i,\gamma_0,\theta_0)\mid X_i\right]=0
\qquad\text{a.s.},
\label{CMR}
\end{equation}
where $\varepsilon:\mathcal W\times\Gamma\times\Theta\rightarrow\mathbb R^q$ is a known measurable moment function, $W_i$ denotes the observed data, $\gamma_0\in\Gamma$ is a possibly infinite-dimensional nuisance parameter, and $\theta_0\in\Theta\subset\mathbb R^p$ is a finite-dimensional parameter. The observations $\{W_i\}_{i=1}^n$ are independent and identically distributed (iid), and $X_i$ is a subvector of $W_i$ taking values in $\mathcal X\subset\mathbb R^{d_x}$, with $d_x$ fixed. Although the theory extends to vector-valued moments, we restrict attention to the scalar case, $q=1$, to simplify the exposition.

Conditional moment restrictions encompass conditional mean and generalized linear models, quantile regression, instrumental-variable models, and many other semiparametric models. Our objective is to test \eqref{CMR} against unrestricted nonparametric alternatives when the nuisance parameters may be estimated by machine learning.

A distinctive feature of our framework is that the nuisance parameters are defined as population functionals under both correct specification and misspecification. Let $\mathcal F$ be a collection of probability distributions for $W_i$, which is unrestricted except for regularity conditions such as the existence of certain moments and derivatives, and let $F_0\in\mathcal F$ denote the true distribution. Expectations under $F\in\mathcal F$ are denoted by $E_F[\cdot]$, while expectations under $F_0$ are written as $E[\cdot]$.

\subsection{Population Nuisance Functionals}

For each $F\in\mathcal F$, define $\gamma_F\in\Gamma$ by
\begin{equation}
E_F\!\left[\delta(Z_i)\rho(W_i,\gamma_F)\right]=0,
\qquad
\forall\,\delta\in\Gamma,
\label{gamma_ident}
\end{equation}
where $\rho(W_i,\gamma_F)$ is a generalized error and $\Gamma$ is a closed linear subspace of $L_2(F_Z)$. Henceforth, $L_2(F)$ denotes the space of square-integrable functions with respect to $F$, equipped with the norm $\|f\|_{L_2(F)}=\{E_F[f^2]\}^{1/2}.$ Here $F_Z$ denotes the marginal distribution of $Z_i$ induced by $F$. The formulation in (\ref{gamma_ident}) covers fully nonparametric, additive, high-dimensional linear, generalized linear, and quantile regression models; see Ichimura and Newey (2022). The generalized error $\rho$ and conditioning variable $Z_i$ need not coincide with the testing moment $\varepsilon$ and conditioning variable $X_i$, so nuisance estimation may use information beyond the restriction being tested.

The finite-dimensional parameter $\theta_F\in\Theta$ solves
\begin{equation}
E_F\!\left[
A_F(X_i)\varepsilon(W_i,\gamma_F,\theta_F)
\right]=0,
\label{theta_ident}
\end{equation}
where $A_F(X_i)\in\mathbb R^p$ is an identified, possibly unknown, column-vector instrument. Equations \eqref{gamma_ident} and \eqref{theta_ident} define $\eta_F=(\gamma_F,\theta_F)\in\Gamma\times\Theta,
$
which specializes to $\eta_0=(\gamma_0,\theta_0)$ when $F=F_0$.

\subsection{The Testing Problem}

Define the null model by
\[
\mathcal F_0
=
\left\{
F\in\mathcal F:
E_F\!\left[\varepsilon(W_i,\eta_F)\mid X_i\right]=0
\;\text{a.s.}
\right\},
\]
where $\eta_F$ is defined by \eqref{gamma_ident}--\eqref{theta_ident}, and assume that $\mathcal F_0$ is nonempty. The testing problem is
\[
H_0:\ F_0\in\mathcal F_0,
\qquad
H_1:\ F_0\in\mathcal F\setminus\mathcal F_0.
\]
Because $\eta_F$ is defined for every $F\in\mathcal F$, the same population formulation applies under the null and alternative hypotheses. This feature is important for the orthogonalization and local-power analysis below.

Henceforth, we omit ``almost surely'' when no confusion can arise. For a vector or matrix $A$, $|A|$ denotes its Euclidean norm and $A'$ its transpose. For real Hilbert spaces $\mathcal H_1$ and $\mathcal H_2$ and a bounded linear operator $T:\mathcal H_1\to\mathcal H_2$, let $T':\mathcal H_2\to\mathcal H_1$ denote its adjoint, defined by
\[
\langle Th_1,h_2\rangle_{\mathcal H_2}
=
\langle h_1,T'h_2\rangle_{\mathcal H_1},
\qquad
h_1\in\mathcal H_1,\quad h_2\in\mathcal H_2.
\]

\subsection{Motivating Examples}

\paragraph{Example 1: High-dimensional mean and generalized linear models.}

Suppose
\[
E\!\left[Y_i-\Lambda(\gamma_0(X_i))\mid X_i\right]=0,
\]
where $\Lambda$ is a known link function, including the identity and logistic links. Let $\{b_j(\cdot)\}_{j=1}^{\infty}$ be a dictionary of functions of $X_i$ and let $\Gamma$ be the mean-square closure of their finite linear span. Then \eqref{gamma_ident} applies with $Z_i=X_i$ and $\rho(W_i,\gamma_F)=Y_i-\Lambda(\gamma_F(X_i)).$ The nuisance function may be estimated using high-dimensional or machine-learning methods, including Lasso; see Tibshirani (1996) and B\"uhlmann and van de Geer (2011). Relative to Escanciano (2024), we allow the nuisance dimension to increase with the sample size and permit $\Gamma$ to be infinite-dimensional. Concurrent and independent work by Gaio et al. (2026) develops related
tests for generalized partially linear models, corresponding to a
particular choice of $\Gamma$ here, using local-polynomial estimation of the nonparametric component.

\paragraph{Example 2: Significance testing.}

Let $X_i=(Z_i,D_i)$ and suppose
\[
E\!\left[Y_i-\gamma_0(Z_i)\mid Z_i,D_i\right]=0.
\]
Under the null, $D_i$ has no predictive power for $Y_i$ conditional on $Z_i$. Classical tests typically estimate $\gamma_0$ nonparametrically in low-dimensional settings; see Delgado and Gonz\'alez-Manteiga (2001). Our framework instead permits generic machine-learning estimation of a high- or infinite-dimensional $\gamma_0$.

\paragraph{Example 3: Constant conditional treatment effects.}

Let $Y_i(1)$ and $Y_i(0)$ be the potential outcomes, $D_i\in\{0,1\}$ the treatment indicator, and $X_i$ a vector of pre-treatment covariates, with
\[
W_i=(Y_i,X_i,D_i),
\qquad
Y_i=D_iY_i(1)+(1-D_i)Y_i(0).
\]
The conditional average treatment effect (CATE) and conditional average treatment effect on the treated (CATT) are
\[
\tau_{\mathrm{CATE}}(X_i)
=
E\!\left[Y_i(1)-Y_i(0)\mid X_i\right],
\qquad
\tau_{\mathrm{CATT}}(X_i)
=
E\!\left[Y_i(1)-Y_i(0)\mid X_i,D_i=1\right].
\]

Under unconfoundedness and overlap (see Rosenbaum and Rubin, 1983), these effects are identified through doubly robust scores involving
\[
\mu_d(X_i)=E[Y_i\mid X_i,D_i=d],
\qquad d\in\{0,1\},
\qquad
e_0(X_i)=P(D_i=1\mid X_i).
\]
Let $\gamma_0=(\mu_0,\mu_1,e_0)$ and define
\[
U_i(\gamma_0)
=
\mu_1(X_i)-\mu_0(X_i)
+
\frac{D_i\{Y_i-\mu_1(X_i)\}}{e_0(X_i)}
-
\frac{(1-D_i)\{Y_i-\mu_0(X_i)\}}{1-e_0(X_i)}
\]
and
\[
V_i(\gamma_0)
=
\left\{
D_i-(1-D_i)\frac{e_0(X_i)}{1-e_0(X_i)}
\right\}
\{Y_i-\mu_0(X_i)\};
\]
see Robins et al. (1994). Constant CATE and CATT can then be expressed as conditional moment restrictions with
\[
\varepsilon(W_i,\eta_0)
=
U_i(\gamma_0)-\theta_0,
\qquad
\varepsilon(W_i,\eta_0)
=
V_i(\gamma_0)-\theta_0D_i,
\]
respectively. The nuisance functions may be estimated by generic machine-learning methods.

This example is related to testing and learning treatment-effect
heterogeneity; see, for example, Crump et al. (2008), Athey et al. (2019), and Kennedy (2023). Our objective is not to estimate heterogeneous effects, but to test constancy while allowing machine-learning estimation of the nuisance functions. Concurrent and independent work by Lu and Song (2026) and Lapenta et al. (2026) develops orthogonal ICM tests specifically for treatment-effect heterogeneity. In contrast, treatment-effect homogeneity is one application of our general RKHS framework for conditional moment specification testing.

\section{Orthogonal Empirical Processes}
\subsection{Influence-Function Representation}
The conditional moment restriction \eqref{CMR} implies
\[
E\!\left[k(X_i)\varepsilon(W_i,\eta_0)\right]=0,
\qquad
\forall\,k\in\mathcal K,
\]
where $\eta_0=(\gamma_0,\theta_0)$ and $\mathcal K$ is a class of test functions. We take $\mathcal K$ to be the unit ball of an RKHS $\mathcal H_K$ on $\mathcal X$ with reproducing kernel $K:\mathcal X\times\mathcal X\rightarrow\mathbb R$,
\[
\mathcal K
=
\left\{
k\in\mathcal H_K:\|k\|_K\leq1
\right\},
\]
where $\langle\cdot,\cdot\rangle_K$ and $\|\cdot\|_K$ denote the RKHS inner product and norm.

Our objective is to construct an empirical process whose first-order behavior is insensitive to nuisance estimation. Consider the adjusted moment
\begin{equation}
\psi(w,\eta,\alpha,k)
=
k(x)\varepsilon(w,\eta)
+
\phi(w,\eta,\alpha,k),
\label{orthogonal moments}
\end{equation}
where $\alpha$ is an additional nuisance parameter. We require the adjustment to satisfy the nonparametric unbiasedness condition $E_F[\phi(W,\eta_F,\alpha_F,k)]=0$ for $F\in\mathcal F$, $k\in\mathcal K$; that $\psi$ be Neyman orthogonal with respect to $(\eta,\alpha)$ at $(\eta_0,\alpha_0)$; and that $k\mapsto\psi(w,\eta,\alpha,k)$ be linear and continuous on $\mathcal H_K$.

To construct the adjustment, consider the paths
\begin{equation}
F_\tau=(1-\tau)F_0+\tau H,
\qquad
\tau\in[0,1],
\label{paths}
\end{equation}
where $H$ belongs to a collection $\mathcal H$ of probability distributions for which $\eta_{F_\tau}$ exists for sufficiently small $\tau$. Suppose that
\begin{align}
\left.
\frac{d}{d\tau}
E\!\left[
k(X_i)\varepsilon(W_i,\eta_{F_\tau})
\right]
\right|_{\tau=0}
&=
\int
\phi(w,\eta_0,\alpha_0,k)\,H(dw),
\label{infdef}
\\
E_F[\phi(W,\eta_F,\alpha_F,k)]
&=0,
\qquad
E[\phi(W,\eta_0,\alpha_0,k)^2]<\infty,
\nonumber
\end{align}
for every $H\in\mathcal H$, $F\in\mathcal F$, and $k\in\mathcal K$. We call $k\mapsto\phi(w,\eta,\alpha,k)$ the \emph{first-step influence process} (FSIP). It extends the usual influence-function representation to an RKHS-indexed process and provides the adjustment to orthogonalize the entire collection of testing moments.

\subsection{Construction of the First-Step Influence Process}
The FSIP corrects separately for estimation of the infinite-dimensional nuisance function $\gamma_0$ and the finite-dimensional parameter $\theta_0$. Define
\[
\dot\varepsilon_\theta(w,\gamma,\theta)
=
\frac{\partial\varepsilon(w,\gamma,\theta)}{\partial\theta'}
\in\mathbb R^{1\times p},
\]
\[
G_0(k)
=
E\!\left[
k(X_i)\dot\varepsilon_\theta(W_i,\gamma_0,\theta_0)
\right],
\qquad
J_0
=
E\!\left[
A_{F_0}(X_i)
\dot\varepsilon_\theta(W_i,\gamma_0,\theta_0)
\right].
\]

\smallskip
\noindent
\textbf{Assumption 1.} \textbf{(a)} $\{W_i\}_{i=1}^n$ is a sample of iid observations; \textbf{(b)} The kernel $K$ is bounded and continuous; \textbf{(c)} For every $k\in\mathcal H_K$ and, when $\theta_0$ is present, componentwise for $k=A_{F_0}$, there exists a function $V_\gamma(\cdot,k)$, independent of $H$, such that $E[V_\gamma(W_i,k)^2]<\infty$ and, for every $H\in\mathcal H$,
\[
\left.
\frac{d}{d\tau}
E\!\left[
k(X_i)\varepsilon(W_i,\gamma_{F_\tau},\theta_0)
\right]
\right|_{\tau=0}
=
\left.
\frac{d}{d\tau}
E\!\left[
V_\gamma(Z_i,k)\gamma_{F_\tau}(Z_i)
\right]
\right|_{\tau=0}.
\]
Moreover, $k\mapsto V_\gamma(z,k)$ is linear for every $z$;

\smallskip
\noindent
\textbf{(d)} For every $w\in\mathcal W$, the function $\varepsilon(w,\gamma_0,\theta)$ is continuously differentiable in $\theta$ on a neighborhood $\Theta_0$ of $\theta_0$, and
\[
E\!\left[
\varepsilon(W_i,\gamma_0,\theta_0)^2
\right]
+
E\!\left[
\sup_{\theta\in\Theta_0}
\left|
\dot\varepsilon_\theta(W_i,\gamma_0,\theta)
\right|^2
\right]
<\infty;
\]

\smallskip
\noindent
\textbf{(e)} There exists a bounded measurable function $V_\rho(Z_i)<0$ such that, for every $H\in\mathcal H$ and $\delta\in\Gamma$,
\[
\left.
\frac{d}{d\tau}
E[\delta(Z_i)\rho(W_i,\gamma_{F_\tau})]
\right|_{\tau=0}
=
\left.
\frac{d}{d\tau}
E[\delta(Z_i)V_\rho(Z_i)\gamma_{F_\tau}(Z_i)]
\right|_{\tau=0}.
\]
Moreover, for every $k$ specified in (c), the weighted projection in \eqref{r_g} exists and is unique;

\smallskip
\noindent
\textbf{(f)} The identifying equations \eqref{gamma_ident} and \eqref{theta_ident} hold at $F=F_\tau$ for every $\tau\in[0,1]$ and are differentiable at $\tau=0$ for every $H\in\mathcal H$. Furthermore,
\[
\left.
\frac{d}{d\tau}
E\!\left[
A_{F_\tau}(X_i)\varepsilon(W_i,\eta_0)
\right]
\right|_{\tau=0}
=0,
\]
$E[|A_{F_0}(X_i)|^2]<\infty$, and $J_0$ is nonsingular.
\smallskip

Assumption~1 imposes smoothness, integrability, and identification conditions. Assumptions~1(c) and 1(e) characterize the pathwise derivatives of the testing moment and nuisance-identifying equation, respectively. Boundedness of $K$ implies $|k(X_i)|^2\leq K(X_i,X_i)\|k\|_K^2$, so $G_0(k)$ is well defined under Assumptions~1(b) and 1(d). 

For each $k\in\mathcal K$, define the correction for estimation of $\gamma_0$ by
\begin{equation}
\alpha_{0\gamma}(\cdot,k)
=
\arg\min_{\alpha\in\Gamma}
E\!\left[
-V_\rho(Z_i)
\left\{
-\frac{V_\gamma(Z_i,k)}{V_\rho(Z_i)}
-\alpha(Z_i)
\right\}^2
\right].
\label{r_g}
\end{equation}
When $\rho(W_i,\gamma_F)=Y_i-\gamma_F(Z_i)$, $V_\rho(Z_i)\equiv-1$, and therefore
\[
\alpha_{0\gamma}(z,k)
=
\Pi_\Gamma\!\left(V_\gamma(\cdot,k)\right)(z),
\]
where $\Pi_\Gamma$ is the orthogonal projection onto $\Gamma$. The correction for estimation of $\theta_0$ is
\[
\alpha_{0\theta}(x,k)
=
G_0(k)J_0^{-1}A_{F_0}(x),
\qquad\text{so}\qquad
\widetilde k(x)
=
k(x)-\alpha_{0\theta}(x,k)
\]
is the adjusted test function. Linearity of $k\mapsto V_\gamma(z,k)$ implies
\[
\alpha_{0\gamma}(z,\widetilde k)
=
\alpha_{0\gamma}(z,k)
-
G_0(k)J_0^{-1}
\alpha_{0\gamma}(z,A_{F_0}),
\]
where $\alpha_{0\gamma}(z,A_{F_0})$ is interpreted componentwise. The resulting adjustment is
\begin{equation}
\phi(w,\eta_0,\alpha_0,k)
=
\alpha_{0\gamma}(z,\widetilde k)\rho(w,\gamma_0)
-
\alpha_{0\theta}(x,k)\varepsilon(w,\eta_0),
\label{phi_final}
\end{equation}
and hence the orthogonal moment is
\begin{equation}
\psi(w,\eta_0,\alpha_0,k)
=
\widetilde k(x)\varepsilon(w,\eta_0)
+
\alpha_{0\gamma}(z,\widetilde k)\rho(w,\gamma_0),
\label{FSIFform}
\end{equation}
where $\alpha_0=(\alpha_{0\gamma},\alpha_{0\theta})$.

\begin{proposition}[Construction of the FSIP]
\label{prop:FSIP}
Suppose Assumption~1 holds. Then the adjustment $\phi$ in \eqref{phi_final} satisfies the influence-function representation \eqref{infdef}, and the moment $\psi$ in \eqref{FSIFform} is Neyman orthogonal:
\[
\left.
\frac{d}{d\tau}
E\!\left[
\psi(W_i,\eta_{F_\tau},\alpha_{F_\tau},k)
\right]
\right|_{\tau=0}
=0
\]
for every $H\in\mathcal H$ and $k\in\mathcal K$.
\end{proposition}

\subsection{FSIP in the Motivating Examples}
We now specialize \eqref{FSIFform} to the examples in Section~2.

\paragraph{Example 1: High-dimensional mean and generalized linear models.}
Here $\varepsilon\equiv\rho=Y-\Lambda(\gamma_0(X))$ does not depend on $\theta$, so $\alpha_{0\theta}\equiv0$ and $\widetilde k=k$. Let $\lambda(u)=\partial\Lambda(u)/\partial u$ be positive and bounded. Then
\[
V_\rho(x)
=
-\lambda(\gamma_0(x)),
\qquad
V_\gamma(x,k)
=
-k(x)\lambda(\gamma_0(x)).
\]
Define the weighted projection
\begin{equation}
\Pi_{\Gamma,\lambda}k
=
\arg\min_{g\in\Gamma}
E\!\left[
\lambda(\gamma_0(X_i))
\{k(X_i)-g(X_i)\}^2
\right].
\label{projection}
\end{equation}
It follows that $\alpha_{0\gamma}(\cdot,k)=-\Pi_{\Gamma,\lambda}k$, and
\[
\psi(w,\gamma_0,\alpha_0,k)
=
\left[
k(x)-\Pi_{\Gamma,\lambda}k(x)
\right]
\varepsilon(w,\gamma_0).
\]

\paragraph{Example 2: Significance testing.}
Again, $\varepsilon\equiv\rho=Y-\gamma_0(Z)$ does not depend on $\theta$, so $\alpha_{0\theta}\equiv0$ and $\widetilde k=k$. Since $V_\rho(z)=-1$ and $V_\gamma(z,k)=-E[k(X_i)\mid Z_i=z]$, with $\Gamma=L_2(F_Z)$, we obtain $\alpha_{0\gamma}(\cdot,k)=-E[k(X_i)\mid Z_i=\cdot]$. Therefore,
\[
\psi(w,\gamma_0,\alpha_0,k)
=
\left[
k(x)-E[k(X_i)\mid Z_i=z]
\right]
\varepsilon(w,\gamma_0).
\]

\paragraph{Example 3: Constant conditional treatment effects.}
The doubly robust scores are already orthogonal with respect to $\gamma_0=(\mu_1,\mu_0,e_0)$, so $\alpha_{0\gamma}\equiv0$. Moreover, $A_F\equiv1$, and
\[
\widetilde k(X_i)
=
k(X_i)
-
\frac{E[k(X_i)L_i]}{E[L_i]},
\]
where $L_i=1$ for CATE and $L_i=D_i$ for CATT. Hence, $\psi(w,\eta_0,\alpha_0,k)=\widetilde k(X_i)\varepsilon(w,\eta_0)$.

All three examples admit the representation
\[
\psi(W_i,\eta_0,\alpha_0,k)
=
(\Pi_0^\perp k)(X_i)\varepsilon(W_i,\eta_0),
\]
where $\Pi_0^\perp$ maps $\mathcal H_K$ into functions on $\mathcal X$, but need not map $\mathcal H_K$ into itself.
\section{Kernel-Based Locally Robust Tests}
\subsection{Cross-Fitted Orthogonal Process}
The orthogonal process developed in Section~3 depends on unknown nuisance functions. We construct a feasible version using cross-fitting, which permits flexible machine-learning estimators while avoiding empirical-process restrictions associated with estimating and evaluating the nuisance functions on the same observations.

Partition $\{1,\ldots,n\}$ into $L<\infty$ folds $\mathcal N=\{I_\ell\}_{\ell=1}^L$, and let $I_\ell^c=\{1,\ldots,n\}\setminus I_\ell$ denote the corresponding training sample. For each fold $\ell$, construct $\hat\eta_\ell=(\hat\gamma_\ell,\hat\theta_\ell)$, $\hat A_\ell$, and $\hat\alpha_{\gamma,\ell}$ using only observations in $I_\ell^c$. For $i\in I_\ell$, define $\hat\varepsilon_{i,\ell}=\varepsilon(W_i,\hat\eta_\ell)$ and $\hat\rho_{i,\ell}=\rho(W_i,\hat\gamma_\ell)$.

The sample analogues of $G_0$ and $J_0$ are full-sample averages of cross-fitted evaluations,
\[
\hat G_n(k)
=
\frac{1}{n}
\sum_{\ell=1}^L
\sum_{i\in I_\ell}
k(X_i)
\dot\varepsilon_\theta
(W_i,\hat\gamma_\ell,\hat\theta_\ell),
\qquad
\hat J_n
=
\frac{1}{n}
\sum_{\ell=1}^L
\sum_{i\in I_\ell}
\hat A_\ell(X_i)
\dot\varepsilon_\theta
(W_i,\hat\gamma_\ell,\hat\theta_\ell).
\]
Define
\begin{align*}
\hat k_\ell(x)
&=
k(x)
-
\hat G_n(k)\hat J_n^{-1}\hat A_\ell(x),
\\
\hat\psi_{i,\ell}(k)
&=
\hat k_\ell(X_i)\hat\varepsilon_{i,\ell}
+
\hat\alpha_{\gamma,\ell}
(Z_i,\hat k_\ell)\hat\rho_{i,\ell}.
\end{align*}
The feasible orthogonal empirical process is
\[
\hat\nu_n(k)
=
\frac{1}{\sqrt n}
\sum_{\ell=1}^L
\sum_{i\in I_\ell}
\hat\psi_{i,\ell}(k),
\qquad
k\in\mathcal K.
\]

\subsection{Locally Robust Kernel Statistic}
The following condition ensures that the feasible process admits an RKHS representation.

\medskip
\noindent
\textbf{Assumption 2 (RKHS representation).} For every fold $\ell$, there exist $r_{\gamma,0}(Z_i,\cdot),\hat r_{\gamma,\ell}(Z_i,\cdot)\in\mathcal H_K$ such that, for every $k\in\mathcal H_K$, $\alpha_{0\gamma}(Z_i,k)
=
\langle r_{\gamma,0}(Z_i,\cdot),k\rangle_K,
$ and $\hat\alpha_{\gamma,\ell}(Z_i,k)
=
\langle\hat r_{\gamma,\ell}(Z_i,\cdot),k\rangle_K.
$

\medskip
\noindent Define
\begin{align*}
\hat g_n(\cdot)
&=
\frac{1}{n}
\sum_{\ell=1}^L
\sum_{i\in I_\ell}
K(X_i,\cdot)
\dot\varepsilon_\theta
(W_i,\hat\gamma_\ell,\hat\theta_\ell)',
\qquad
\hat C_n(\cdot)'
=
\hat g_n(\cdot)'\hat J_n^{-1},
\\
\hat B_\ell(x,\cdot)
&=
K(x,\cdot)
-
\hat C_n(\cdot)'\hat A_\ell(x).
\end{align*}
Also define $\widetilde{\hat r}_{\gamma,\ell}(z,\cdot)=\hat r_{\gamma,\ell}(z,\cdot)-\hat C_n(\cdot)'\hat\alpha_{\gamma,\ell}(z,\hat A_\ell)$, where the last term is interpreted componentwise. By linearity and the reproducing property,
\[
\hat\psi_{i,\ell}(k)
=
\langle\hat r_{i,\ell},k\rangle_K,
\qquad
\hat r_{i,\ell}
=
\hat\varepsilon_{i,\ell}\hat B_\ell(X_i,\cdot)
+
\hat\rho_{i,\ell}
\widetilde{\hat r}_{\gamma,\ell}(Z_i,\cdot).
\]
Consequently, $\hat\nu_n(k)=\langle\hat R_n,k\rangle_K$ with $\hat R_n=n^{-1/2}\sum_{\ell=1}^L\sum_{i\in I_\ell}\hat r_{i,\ell}$, and the locally robust kernel statistic is
\begin{equation}
LRK_n
=
\sup_{\|k\|_K\leq1}
\hat\nu_n(k)^2
=
\|\hat R_n\|_K^2
=
\frac{1}{n}
\mathbf 1_n'
\hat{\mathbf R}
\mathbf 1_n,
\label{LRKstat}
\end{equation}
where $\mathbf 1_n=(1,\ldots,1)'$ and $\hat{\mathbf R}=(\hat R_{ij})_{i,j=1}^n$, $\hat R_{ij}=\langle\hat r_{i,\ell},\hat r_{j,m}\rangle_K$ for $i\in I_\ell$, $j\in I_m$. Thus, the infinite-dimensional supremum reduces to a quadratic-form matrix calculation.

\subsection{Multiplier Bootstrap}
The null distribution of $LRK_n$ depends on the unknown data-generating distribution. Let $\mathbf V=(V_1,\ldots,V_n)'$ contain iid multipliers, independent of the data, with zero mean and unit variance. In the simulations and application, we use the two-point distribution of Mammen (1993):
\[
V_i
=
\begin{cases}
(1+\sqrt 5)/2,
&
\text{with probability }
(\sqrt 5-1)/(2\sqrt 5),
\\[0.3em]
(1-\sqrt 5)/2,
&
\text{with probability }
(\sqrt 5+1)/(2\sqrt 5).
\end{cases}
\]
Conditional on the data, define $LRK_n^*=n^{-1}\mathbf V'\hat{\mathbf R}\mathbf V$. For $0<\varsigma<1$, let $c_{1-\varsigma}^*$ be the conditional $(1-\varsigma)$ quantile of $LRK_n^*$. The test rejects $H_0$ when $LRK_n>c_{1-\varsigma}^*$. The bootstrap does not require nuisance re-estimation; its asymptotic validity is established below in Theorem \ref{thm:bootstrap}.

\subsection{Implementation in the Motivating Examples}

Let $\ell(i)$ denote the fold containing observation $i$ and, for
brevity, set
$\hat\varepsilon_i=\hat\varepsilon_{i,\ell(i)}$, where
$\hat\varepsilon_{i,\ell}=\varepsilon(W_i,\hat\eta_\ell)$ was defined
above. In the three examples,
\[
\psi(W_i,\eta_0,\alpha_0,k)
=
\varepsilon(W_i,\eta_0)(\Pi_0^\perp k)(X_i),
\qquad
\hat\psi_{i,\ell}(k)
=
\hat\varepsilon_{i,\ell}(\hat\Pi_\ell^\perp k)(X_i),
\]
where $\Pi_0^\perp$ and $\hat\Pi_\ell^\perp$ remove the relevant
nuisance component of $k$.

Write $\mathbf v=(k(X_1),\ldots,k(X_n))'$ and let
$\hat{\mathbf P}$ denote the example-specific nuisance-adjustment
matrix. With $\hat{\mathbf M}=\mathbf I_n-\hat{\mathbf P}$,
\[
(\hat\Pi_{\ell(i)}^\perp k)(X_i)
=
(\hat{\mathbf M}\mathbf v)_i.
\]
The corresponding representer is
\[
\hat\Phi_i
=
K(X_i,\cdot)
-
\sum_{j=1}^n\hat P_{ij}K(X_j,\cdot),
\]
so that
$\hat r_i=\hat\varepsilon_i\hat\Phi_i$ and
\[
\hat{\mathbf K}^{\perp}
=
\hat{\mathbf M}\mathbf K\hat{\mathbf M}',
\qquad
\mathbf K=(K(X_i,X_j))_{i,j=1}^n.
\]
Hence,
\begin{equation}
LRK_n
=
\frac{1}{n}
\hat{\boldsymbol\varepsilon}'
\hat{\mathbf K}^{\perp}
\hat{\boldsymbol\varepsilon},
\qquad
\hat{\boldsymbol\varepsilon}
=
(\hat\varepsilon_1,\ldots,\hat\varepsilon_n)'.
\label{LRKmatrix}
\end{equation}

The three examples differ only in the construction of
$\hat{\mathbf P}$. In Example~1, for $i\in I_\ell$,
\[
(\hat{\mathbf P}\mathbf v)_i
=
\hat b_i'
\bigl(
\hat{\mathbf B}_{-\ell}'
\hat{\boldsymbol\Lambda}_{-\ell}
\hat{\mathbf B}_{-\ell}
\bigr)^{-}
\hat{\mathbf B}_{-\ell}'
\hat{\boldsymbol\Lambda}_{-\ell}
\mathbf v_{-\ell},
\]
where $\mathbf v_{-\ell}=(k(X_j):j\in I_\ell^c)'$,
$\hat b_i$ is the sample-standardized dictionary evaluated at $X_i$,
$\hat{\mathbf B}_{-\ell}$ is its training-sample design matrix,
\[
\hat{\boldsymbol\Lambda}_{-\ell}
=
\operatorname{diag}
\{\lambda(\hat\gamma_\ell(X_j)):j\in I_\ell^c\},
\]
and $(\cdot)^{-}$ denotes a generalized inverse.

In Example~2, $\hat\gamma_\ell$ estimates
$\gamma_0(z)=E[Y_i\mid Z_i=z]$, whereas the weights
$\hat\omega_{\ell j}(z)$ estimate the separate conditional-mean
operator
\[
m_0(z;k)=E[k(X_i)\mid Z_i=z],
\qquad
\hat m_\ell(z;k)
=
\sum_{j\in I_\ell^c}
\hat\omega_{\ell j}(z)k(X_j).
\]
Thus, for $i\in I_\ell$,
\[
(\hat{\mathbf P}\mathbf v)_i
=
\hat m_\ell(Z_i;k)
=
\sum_{j\in I_\ell^c}
\hat\omega_{\ell j}(Z_i)v_j.
\]
The weights are constructed using $I_\ell^c$ and do not depend on
$\mathbf v$, preserving linearity in $k$.

In Example~3,
\[
\hat{\mathbf P}
=
\frac{\mathbf 1_n\mathbf L'}{\mathbf L'\mathbf 1_n},
\qquad
\mathbf L=(L_1,\ldots,L_n)',
\]
where $L_i=1$ for CATE and $L_i=D_i$ for CATT. Therefore,
\[
(\hat{\mathbf P}\mathbf v)_i
=
\frac{\mathbf L'\mathbf v}{\mathbf L'\mathbf 1_n}.
\]

Further computational details and the corresponding non-cross-fitted
implementations are provided in the supplementary material.

\section{Asymptotic Theory}

We establish uniform equivalence between the feasible process and its oracle counterpart and then derive the asymptotic distribution, local power, bootstrap validity, and consistency of the LRK test. Here $\rightsquigarrow$ denotes convergence in distribution for real-valued variables and Hoffmann--J{\o}rgensen weak convergence for random elements of $\mathcal H_K$; see van der Vaart and Wellner (1996).

\subsection{Oracle Equivalence}

Define
\[
\nu_n^0(k)=\frac{1}{\sqrt n}\sum_{i=1}^n\psi(W_i,\eta_0,\alpha_0,k),
\qquad
\hat\nu_n(k)=\frac{1}{\sqrt n}\sum_{\ell=1}^L\sum_{i\in I_\ell}
\psi(W_i,\hat\eta_\ell,\hat\alpha_\ell,k).
\]
We consider 
\[
F_n=(1-n^{-1/2})F_0+n^{-1/2}H,  
\]
where $H\in\mathcal H$ has the same $X$-marginal $F_X$ as $F_0$ and
\[
a(\cdot)=E_H[\varepsilon(W_i,\eta_0)\mid X_i=\cdot]\in L_2(F_X).
\]
Since $F_0\in\mathcal F_0$,
\begin{equation}
E_{F_n}[\varepsilon(W_i,\eta_0)\mid X_i]
=n^{-1/2}a(X_i).
\label{local_alter}
\end{equation}

\noindent Write $\varepsilon_i=\varepsilon(W_i,\eta_0)$, $\rho_i=\rho(W_i,\gamma_0)$ and 
\[
C_0(\cdot)'=g_0(\cdot)'J_0^{-1},\quad
B_0(x,\cdot)=K(x,\cdot)-C_0(\cdot)'A_{F_0}(x).
\]
When $\theta_0$ is present, set $\zeta_0(z)=\alpha_{0\gamma}(z,A_{F_0}),$ $\hat\zeta_\ell(z)=\hat\alpha_{\gamma,\ell}(z,\hat A_\ell),
$ componentwise, and
\[
\Delta_{\gamma,\ell}(z,\cdot)
=
\hat r_{\gamma,\ell}(z,\cdot)-r_{\gamma,0}(z,\cdot)
-C_0(\cdot)'\{\hat\zeta_\ell(z)-\zeta_0(z)\}.
\]
When $\theta_0$ is absent, let
$\Delta_{\gamma,\ell}=\hat r_{\gamma,\ell}-r_{\gamma,0}$.
Also define
\[
\widetilde r_{\gamma,0}(z,\cdot)
=
r_{\gamma,0}(z,\cdot)-C_0(\cdot)'\alpha_{0\gamma}(z,A_{F_0}).
\]

For $h=(h_1,\ldots,h_p)'\in\mathcal H_K^p$, let $\|h\|_{K,p}^2=\sum_{j=1}^p\|h_j\|_K^2.$ For a scalar, vector, or RKHS-valued random element $U$, write
\[
\|U\|_{Q,2,\ell}
=
\{E_Q[\|U\|^2\mid I_\ell^c]\}^{1/2},
\qquad Q\in\{F_n,H\},
\]
using the appropriate Euclidean or RKHS norm.

\medskip
\noindent
\textbf{Assumption 3 (Cross-fitting).}
The number of folds $L$ is fixed, and
$\hat\eta_\ell$, $\hat A_\ell$, $\hat\alpha_{\gamma,\ell}$, and
$\hat r_{\gamma,\ell}$ are constructed using only $I_\ell^c$. With probability approaching one, $\alpha_{0\gamma}(\cdot,k),\,
\hat\alpha_{\gamma,\ell}(\cdot,k)\in\Gamma.
$ When $\theta_0$ is present, the same holds componentwise for
$k=A_{F_0}$ and $k=\hat A_\ell$.

\newpage
\noindent
\textbf{Assumption 4 (Mean-square consistency).}
For each $\ell$:

\noindent
(i) $\|\hat\varepsilon_{i,\ell}-\varepsilon_i\|_{F_n,2,\ell}
+\|\hat\rho_{i,\ell}-\rho_i\|_{F_n,2,\ell}
+\|\rho_i\Delta_{\gamma,\ell}(Z_i,\cdot)\|_{F_n,2,\ell}
+\|\rho_i\Delta_{\gamma,\ell}(Z_i,\cdot)\|_{H,2,\ell}
=o_p(1).
$

\noindent
(ii) When $\theta_0$ is present, $\|(1+\varepsilon_i^2)^{1/2}
\{\hat A_\ell(X_i)-A_{F_0}(X_i)\}\|_{F_n,2,\ell}=o_p(1),
$
and
\[
\|(1+\varepsilon_i^2)^{1/2}\hat A_\ell(X_i)\|_{F_n,2,\ell}
+\|\rho_i\hat\zeta_\ell(Z_i)\|_{F_n,2,\ell}
+\|\rho_i\hat\zeta_\ell(Z_i)\|_{H,2,\ell}=O_p(1).
\]

\noindent
(iii) When $\theta_0$ is present, $\|\hat g_n-g_0\|_{K,p}+|\hat J_n-J_0|=o_p(1).
$

\noindent
(iv) $\sup_x\|B_0(x,\cdot)\|_K+
\sup_z\|\widetilde r_{\gamma,0}(z,\cdot)\|_K<\infty.
$

\medskip
\noindent
\textbf{Assumption 5 (Product rates).}
For each $\ell$:

\noindent
(i) When $\theta_0$ is present, $\|\hat\varepsilon_{i,\ell}-\varepsilon_i\|_{F_n,2,\ell}
\|\hat A_\ell(X_i)-A_{F_0}(X_i)\|_{F_n,2,\ell}
=o_p(n^{-1/2}).
$

\noindent
(ii) When $\theta_0$ is present,
\[
\|\hat C_n-C_0\|_{K,p}
\left\{
\|(\hat\varepsilon_{i,\ell}-\varepsilon_i)\hat A_\ell(X_i)\|_{F_n,2,\ell}^2
+
\|(\hat\rho_{i,\ell}-\rho_i)\hat\zeta_\ell(Z_i)\|_{F_n,2,\ell}^2
\right\}^{1/2}
=o_p(n^{-1/2}).
\]

\noindent
(iii) $\|\hat\rho_{i,\ell}-\rho_i\|_{F_n,2,\ell}
\|\Delta_{\gamma,\ell}(Z_i,\cdot)\|_{F_n,2,\ell}
=o_p(n^{-1/2}).
$

\medskip
If $\|\hat C_n-C_0\|_{K,p}=O_p(n^{-1/2})$, Assumption~5(ii) requires only that its braced term be $o_p(1)$. This parametric rate holds in the three motivating examples because
$\dot\varepsilon_\theta(w,\gamma,\theta_0)$ does not depend on $\gamma$.

Define
\[
b_F(\eta)
=
E_F[\psi(W,\eta,\alpha_0,K(X,\cdot))]\in\mathcal H_K.
\]
The reproducing property gives
\[
\|b_F(\eta)-b_F(\eta_0)\|_K
=
\sup_{\|k\|_K\leq1}
|E_F[\psi(W,\eta,\alpha_0,k)]
-E_F[\psi(W,\eta_0,\alpha_0,k)]|.
\]

\medskip
\noindent
\textbf{Assumption 6 (Orthogonality bias).}
For each $\ell$, $\sqrt n\,
\|b_{F_n}(\hat\eta_\ell)-b_{F_n}(\eta_0)\|_K=o_p(1).
$

\medskip

Assumptions~4--6 impose consistency, product-rate, and second-order bias conditions. A sufficient condition for Assumption~6 is
\[
\|b_{F_n}(\eta)-b_{F_n}(\eta_0)\|_K
\lesssim
n^{-1/2}\|\eta-\eta_0\|_\eta+\|\eta-\eta_0\|_\eta^2,
\qquad
\|\hat\eta_\ell-\eta_0\|_\eta=o_p(n^{-1/4}).
\]
Primitive sufficient conditions and their verification in the motivating examples are given in Appendix B of the supplementary material.

\begin{theorem}[Oracle equivalence]
\label{thm:oracleequiv}
Under Assumptions~1--6, under $\{F_n\}$,
\[
\sup_{\|k\|_K\leq1}
|\hat\nu_n(k)-\nu_n^0(k)|=o_p(1).
\]
\end{theorem}

\subsection{Asymptotic Distribution and Local Power}

The oracle moment has Riesz representer
\[
r_0(w,\cdot)
=
\varepsilon(w,\eta_0)B_0(x,\cdot)
+\rho(w,\gamma_0)\widetilde r_{\gamma,0}(z,\cdot).
\]
Thus,
\[
\psi(W_i,\eta_0,\alpha_0,k)
=\langle r_0(W_i,\cdot),k\rangle_K,\qquad
R_n^0=\frac{1}{\sqrt n}\sum_{i=1}^n r_0(W_i,\cdot).
\]

\medskip
\noindent
\textbf{Assumption 7 (Moment condition).}
For some $\delta>0$, $\sup_nE_{F_n}[\|r_0(W_i,\cdot)\|_K^{2+\delta}]<\infty.$
\medskip
\noindent Under Assumption~4(iv), this follows from $\sup_nE_{F_n}[|\varepsilon_i|^{2+\delta}+|\rho_i|^{2+\delta}]<\infty.$

Define
\[
\mu_H=E_H[r_0(W_i,\cdot)]\in\mathcal H_K,
\qquad
\Sigma_0=E[r_0(W_i,\cdot)\otimes r_0(W_i,\cdot)],
\]
where $(h_1\otimes h_2)k=\langle h_2,k\rangle_Kh_1.$

\begin{theorem}[Local asymptotic distribution]
\label{thm:local}
Under Assumptions~1--7 and the local alternatives \eqref{local_alter},
\[
R_n^0\rightsquigarrow\mathbb Z+\mu_H
\qquad\text{in }\mathcal H_K,
\]
where $\mathbb Z$ is centered Gaussian with covariance operator $\Sigma_0$. Consequently,
\[
LRK_n\rightsquigarrow\|\mathbb Z+\mu_H\|_K^2.
\]
Under the null, the limit is $\|\mathbb Z\|_K^2$.
\end{theorem}

We next characterize the detectable local directions.

\medskip
\noindent
\textbf{Assumption 8 (Oracle feature representation).}
There exists a strongly measurable
$\Phi_0:\mathcal X\to\mathcal H_K$ such that $r_0(W_i,\cdot)=\varepsilon_i\Phi_0(X_i,\cdot),
$ and $E[\|\Phi_0(X_i,\cdot)\|_K^2]<\infty,
$ and, for some $C<\infty$,
\[
\|\Pi_0^\perp k\|_{L_2(F_X)}
\leq C\|k\|_{L_2(F_X)},
\qquad
(\Pi_0^\perp k)(x)=\langle\Phi_0(x,\cdot),k\rangle_K.
\]

The adjoint of $\Pi_0^\perp:\mathcal H_K\to L_2(F_X)$ satisfies
\[
(\Pi_0^\perp)'f
=
E[f(X_i)\Phi_0(X_i,\cdot)]
\in\mathcal H_K,
\]
and hence
\[
\mu_a
=
(\Pi_0^\perp)'a
=
E[a(X_i)\Phi_0(X_i,\cdot)].
\]
Thus, $a$ is locally undetectable exactly when $\mu_a=0$.

A bounded kernel $K$ is \emph{integrally strictly positive with respect
to $F_X$} if
\[
\iint
f(x)K(x,x')f(x')\,dF_X(x)\,dF_X(x')
>
0
\]
for every nonzero $f\in L_2(F_X)$; equivalently, $\mathcal H_K$ is
dense in $L_2(F_X)$. Under Assumption~8, $\Pi_0^\perp$ then extends
uniquely to a bounded operator $\Pi_{0,2}^\perp:L_2(F_X)\to L_2(F_X).$

\begin{proposition}[Locally testable alternatives]
\label{prop:testability}
Suppose Assumption~8 holds and $K$ is integrally strictly positive with
respect to $F_X$. Then
\[
\mu_a=0
\quad\Longleftrightarrow\quad
(\Pi_{0,2}^\perp)'a=0.
\]
Consequently, $a$ is locally detectable if and only if
$(\Pi_{0,2}^\perp)'a\neq0$.
\end{proposition}

The extended adjoint has a simple expression in each example. In
Example~1, let $\lambda_0(x)=\lambda(\gamma_0(x)).$ Since $\Pi_{\Gamma,\lambda_0}$ is self-adjoint under the
$\lambda_0$-weighted inner product,
\[
E\!\left[
a(X_i)\Pi_{\Gamma,\lambda_0}k(X_i)
\right]
=
E\!\left[
\lambda_0(X_i)
\Pi_{\Gamma,\lambda_0}
\left(\frac{a}{\lambda_0}\right)(X_i)
k(X_i)
\right].
\]
Therefore,
\[
(\Pi_{0,2}^\perp)'a
=
a
-
\lambda_0
\Pi_{\Gamma,\lambda_0}
\left(\frac{a}{\lambda_0}\right).
\]
In Example~2, conditional expectation is an orthogonal projection in
$L_2(F_X)$, so
\[
(\Pi_{0,2}^\perp)'a
=
a-E[a(X_i)\mid Z_i].
\]
Finally, in Example~3, writing
$\ell_0(X_i)=E[L_i\mid X_i]$, direct calculation gives
\[
(\Pi_{0,2}^\perp)'a
=
a
-
\frac{E[a(X_i)]}{E[L_i]}\ell_0.
\]
For CATE, $\ell_0\equiv1$, whereas for CATT,
$\ell_0=e_0$.

\begin{table}[h]
\centering
\small
\caption{Extended adjoints and locally undetectable directions.}
\label{tab:local_power}
\begin{tblr}{
  colspec={l X[4] X[3]},
  row{1}={font=\bfseries},
  hlines,
  vlines
}
Example
&
$(\Pi_{0,2}^\perp)'a$
&
Undetectable directions
\\
High-dimensional regression
&
$\displaystyle
a-\lambda_0\Pi_{\Gamma,\lambda_0}
\left(a/\lambda_0\right)$
&
$\displaystyle
a=\lambda_0 g,\quad g\in\Gamma$
\\
Significance testing
&
$\displaystyle
a-E[a\mid Z]$
&
$\displaystyle
a=g(Z)$
\\
Constant CATE/CATT
&
CATE: $a-E[a]$

CATT:
$\displaystyle a-e_0\frac{E[a]}{E[e_0]}$
&
CATE: $a=c$

CATT: $a=ce_0$
\\
\end{tblr}
\end{table}

Thus, the LRK test is locally insensitive exactly to directions absorbed
by the model's nuisance components: $\lambda_0g$, $g\in\Gamma$, in
Example~1; functions of $Z$ in Example~2; constants for CATE; and
multiples of $e_0$ for CATT. Every other $L_2(F_X)$ direction is locally
detectable.
\begin{theorem}[Multiplier-bootstrap validity]
\label{thm:bootstrap}
Suppose Assumptions~1--8 hold, $\Sigma_0\neq0$, and the iid multipliers are independent of the data and satisfy
\[
E[V_i]=0,\qquad E[V_i^2]=1,\qquad
E[|V_i|^{2+\delta}]<\infty,
\]
where $\delta$ is as in Assumption~7. Then, under $\{F_n\}$,
\[
\sup_{t\in\mathbb R}
\left|
P^*(LRK_n^*\leq t)
-
P(\|\mathbb Z\|_K^2\leq t)
\right|
\overset{p}{\longrightarrow}0.
\]
If $c_{1-\varsigma}^*$ is the conditional $(1-\varsigma)$ bootstrap quantile and $c_{1-\varsigma}$ is the corresponding quantile of $\|\mathbb Z\|_K^2$, then
\[
c_{1-\varsigma}^*\overset{p}{\longrightarrow}c_{1-\varsigma}.
\]
Moreover,
\[
P(LRK_n>c_{1-\varsigma}^*)\longrightarrow\varsigma
\]
under the null, whereas under the local alternatives,
\[
P(LRK_n>c_{1-\varsigma}^*)
\longrightarrow
P(\|\mathbb Z+\mu_a\|_K^2>c_{1-\varsigma}).
\]
\end{theorem}

Thus, the bootstrap consistently estimates the null distribution without re-estimating the nuisance functions and reproduces the infeasible test's local-power function.

\subsection{Fixed Alternatives}

Suppose $F_0\in\mathcal F\setminus\mathcal F_0$ and $a(\cdot)=E[\varepsilon(W_i,\eta_0)\mid X_i=\cdot]\in L_2(F_X).$ Next result establishes conditions for the consistency of the tests.

\begin{theorem}[Consistency]
\label{thm:consistency}
Suppose Assumptions~2 and 8 and Proposition~\ref{prop:testability} hold, with
\[
(\Pi_{0,2}^\perp)'a\neq0,
\qquad
E[\|r_0(W_i,\cdot)\|_K]<\infty.
\]
Suppose Assumptions~3--5 also hold with $E_{F_n}$ replaced by $E$, the $E_H$ conditions omitted, and the $o_p(n^{-1/2})$ rates in Assumption~5 replaced by $o_p(1)$. If
\[
\frac{1}{n}\sum_{\ell=1}^L\sum_{i\in I_\ell}
\|\hat r_{i,\ell}\|_K^2=O_p(1),
\]
then
\[
\frac{LRK_n}{n}
\overset{p}{\longrightarrow}
\sup_{\|k\|_K\leq1}\Delta(k)^2>0.
\]
Consequently, $LRK_n\overset{p}{\longrightarrow}\infty$, and the LRK test is consistent.
\end{theorem}

The fixed-alternative result requires only consistency and $o_p(1)$ product rates, rather than the $n^{-1/4}$ rates sufficient for local oracle equivalence. If $\Pi_{0,2}^\perp$ is an orthogonal projection, then the testability condition reduces to
$\Pi_{0,2}^\perp a\neq0$.
\section{Monte Carlo and Empirical Application}

We examine finite-sample size, power, feasible-oracle equivalence, and sensitivity to nuisance-space complexity. All tests have nominal level $5\%$ and use five-fold cross-fitting, $999$ Mammen multiplier-bootstrap draws, and $1{,}000$ Monte Carlo replications unless stated otherwise. Rejection frequencies are reported in percent.

\subsection{Design 1: High-Dimensional Nuisance Dictionaries}

The first experiment corresponds to Example~1. Let
\[
X_i=(X_{i1},X_{i2})'\sim N(0,I_2),
\qquad
n\in\{200,400\}.
\]
We use additive and tensor-product Fourier dictionaries containing $J\in\{50,100,150\}$ penalized terms. The dictionaries include the sine terms used in the null DGP, while the tensor-product dictionaries additionally contain interactions. The intercept and linear terms are unpenalized, and the remaining coefficients are estimated by Lasso with
\[
\lambda_n=0.5\sqrt{\frac{\log J}{n}}.
\]
The test uses the Gaussian kernel
\[
K(x,x')
=
\exp\{-\|x-x'\|^2/\sigma^2\},
\]
where $\sigma$ is the sample median of the nonzero pairwise distances. Cross-fitted residuals and, in the logistic model, derivative weights are combined in a pooled empirical projection.

\paragraph{Size under regularization bias.}
Under the null, the linear DGP is
\[
Y_i
=
X_{i1}+X_{i2}
+
0.5\{\sin(X_{i1})+\sin(X_{i2})\}
+
\varepsilon_i,
\qquad
\varepsilon_i\sim N(0,1),
\]
and the logistic DGP is
\[
Y_i\mid X_i
\sim
\operatorname{Bernoulli}
\left[
\Lambda\!\left(
0.5X_{i1}+0.5X_{i2}
+
0.5\{\sin(X_{i1})+\sin(X_{i2})\}
\right)
\right].
\]
The nonzero sine coefficients are penalized, so Lasso shrinkage generates regularization bias even though the null is correctly specified.

We compare the proposed cross-fitted LRK test (LRK-CF), which uses
five-fold cross-fitted residuals and pooled projection weights, with an infeasible oracle using the true residuals and derivative weights and with a non-orthogonal plug-in statistic whose active-set
influence-function bootstrap accounts for regular estimation effects
but not shrinkage bias. Table~\ref{tab:mc-design1-size} shows that LRK-CF remains close to the oracle: rejection frequencies range from $5.7\%$ to $6.6\%$ in the linear model and from $4.0\%$ to $5.8\%$ in the logistic model. The high-dimensional plug-in test over-rejects substantially in the linear model, reaching $17.7\%$. A correctly specified unpenalized low-dimensional plug-in test remains near $5\%$, confirming that the distortion is due to regularization rather than bootstrap failure. The low-dimensional and non-cross-fitted results are reported in the supplementary material.

\begin{table}[!htbp]
\centering
\small
\caption{Design 1: Empirical size under regularization bias.}
\label{tab:mc-design1-size}
\setlength{\tabcolsep}{5pt}
\begin{tabular}{|cc|ccc|ccc|}
\hline
&&
\multicolumn{3}{c|}{Linear}
&
\multicolumn{3}{c|}{Logistic}
\\
\cline{3-8}
$n$ & $J$
& LRK-CF & Oracle & Plug-in
& LRK-CF & Oracle & Plug-in
\\
\hline
200 &  50 & 6.4 & 5.3 & 11.4 & 4.0 & 4.1 & 6.9 \\
200 & 100 & 5.8 & 4.4 & 13.1 & 4.5 & 4.2 & 6.4 \\
200 & 150 & 6.6 & 5.0 & 15.6 & 4.6 & 3.9 & 6.5 \\
400 &  50 & 5.7 & 5.8 & 13.3 & 5.0 & 5.7 & 6.8 \\
400 & 100 & 6.5 & 5.0 & 17.2 & 5.6 & 5.4 & 6.7 \\
400 & 150 & 5.9 & 4.3 & 17.7 & 5.8 & 6.0 & 6.5 \\
\hline
\end{tabular}

\smallskip
\parbox{0.94\linewidth}{\footnotesize
\emph{Notes.} Oracle uses the true residuals and derivative weights. Plug-in is the non-orthogonal statistic with an active-set influence-function bootstrap. The intercept and linear terms are unpenalized; the $J$ Fourier terms are penalized.}
\end{table}

\paragraph{Power and the nuisance space.}
We perturb the baseline regressions along
\[
a(X_i)=\sin(2X_{i1}+2X_{i2}),
\qquad
\delta\in\{0,0.1,0.2,\ldots,1\}.
\]
The linear and logistic DGPs are, respectively,
\[
Y_i=X_{i1}+X_{i2}+\delta a(X_i)+\varepsilon_i
\]
and
\[
Y_i\mid X_i
\sim
\operatorname{Bernoulli}
\left[
\Lambda\!\left(
0.5X_{i1}+0.5X_{i2}+\delta a(X_i)
\right)
\right].
\]
By construction,
\[
a\notin\Gamma_J^{\mathrm{add}},
\qquad
a\in\Gamma_J^{\mathrm{tensor}}.
\]
Thus, the perturbation is a genuine alternative under the additive dictionary but a nuisance direction under the tensor-product dictionary.

Figure~\ref{fig:mc-example1-power}, for $J=100$, shows increasing rejection probabilities under the additive dictionary. At $\delta=1$, LRK-CF rejection in the linear model rises from $75.6\%$ for $n=200$ to $99.5\%$ for $n=400$; in the logistic model it rises from $15.0\%$ to $40.4\%$. Under the tensor-product dictionary, rejection remains near its null value because the perturbation belongs to the nuisance space. Cross-fitted and non-cross-fitted power curves are very similar. Results for $J=50$ and $J=150$ lead to the same conclusions and are reported in the supplementary material. These findings illustrate Proposition~\ref{prop:testability}: LRK detects departures outside the nuisance space and is insensitive to directions within it.

\begin{figure}[!htbp]
\centering
\includegraphics[width=0.88\textwidth]{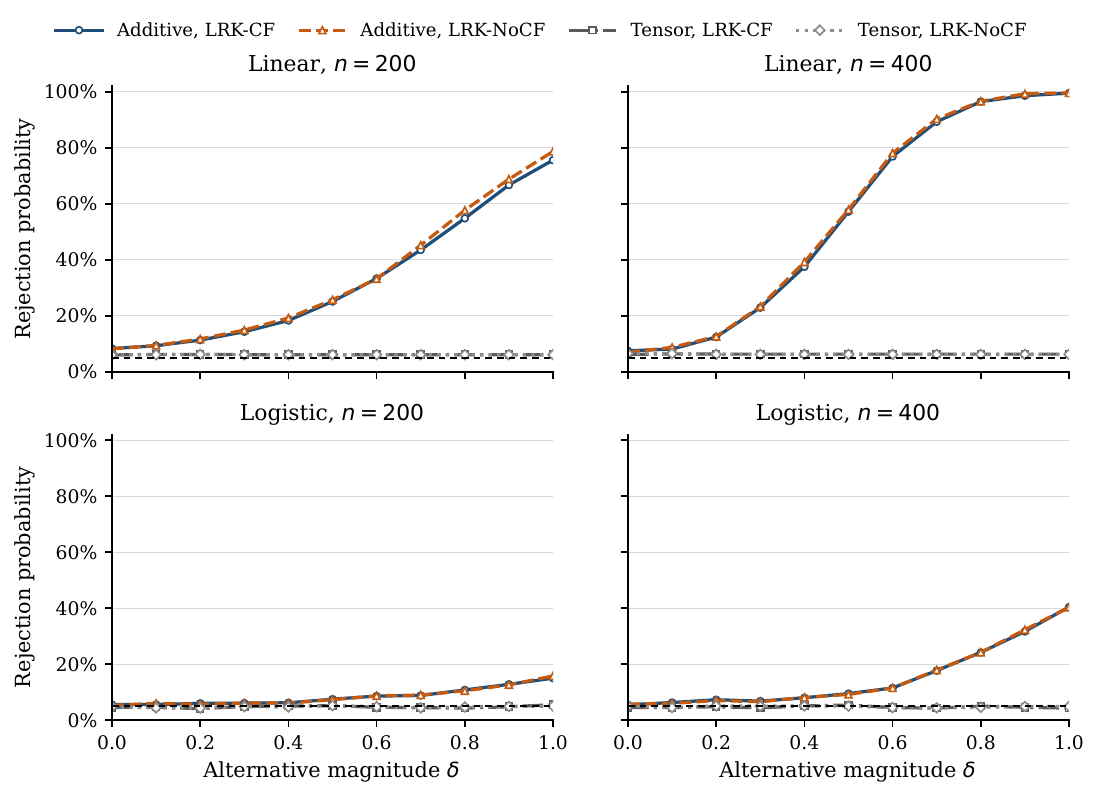}
\caption{Design 1: Rejection probabilities for $J=100$. The perturbation lies outside the additive nuisance space but inside the tensor-product nuisance space. LRK-CF and LRK-NoCF use cross-fitted and full-sample residuals, respectively; the dashed line is the nominal $5\%$ level.}
\label{fig:mc-example1-power}
\end{figure}

\subsection{Design 2: Significance Testing with Machine Learning}

The second experiment corresponds to Example~2. The DGP is
\[
Y_i
=
g_0(Z_i)
+
\frac{\delta}{\sqrt n}D_i
+
\varepsilon_i,
\qquad
g_0(Z_i)=\sin(Z_{i1})+Z_{i2}^2,
\]
where
\[
Z_{ij}\sim U[-1,1],
\qquad
D_i,\varepsilon_i\sim N(0,1),
\]
and all variables are mutually independent. The null is $\delta=0$, while
$\delta\in\{1,2,3,4\}$ gives local alternatives.

The nuisance quantities are estimated by either local-kernel regression or random forests. We consider distance and Gaussian kernels and the kernel induced by the integrated conditional moment statistic of Delgado and Gonz\'alez-Manteiga (2001), denoted D-GM.

Table~\ref{tab:design2_n400} reports results for $n=400$. Size ranges from $3.2\%$ to $5.1\%$ across all learners and kernels, and power increases monotonically with $\delta$. Distance and Gaussian kernels generally outperform D-GM, while results are stable across nuisance learners. Results for $n=200$, reported in the supplementary material, yield the same conclusions.

\begin{table}[!htbp]
\centering
\small
\caption{Design 2: Significance testing with flexible nuisance estimation.}
\label{tab:design2_n400}
\setlength{\tabcolsep}{6pt}
\begin{tabular}{|c|c|ccccc|}
\hline
Learner & Kernel
& $\delta=0$ & $\delta=1$ & $\delta=2$ & $\delta=3$ & $\delta=4$
\\
\hline
Local kernel & Distance & 5.1 & 13.6 & 43.1 & 79.1 & 96.8 \\
Local kernel & Gaussian & 4.7 & 12.0 & 36.3 & 69.7 & 91.7 \\
Local kernel & D-GM     & 3.8 &  9.6 & 29.0 & 57.4 & 84.2 \\
Random forest & Distance & 4.3 & 14.1 & 46.2 & 76.4 & 94.7 \\
Random forest & Gaussian & 4.1 & 12.0 & 37.8 & 68.9 & 90.0 \\
Random forest & D-GM     & 3.2 &  7.2 & 26.7 & 53.9 & 77.5 \\
\hline
\end{tabular}

\smallskip
\parbox{0.88\linewidth}{\footnotesize
\emph{Notes.} $n=400$. Nuisance quantities are estimated by five-fold cross-fitting.}
\end{table}

\subsection{Design 3: Constant CATT}

The third experiment corresponds to Example~3. Let
\[
X_i=(X_{i1},X_{i2})'\sim N(0,I_2),
\qquad
D_i=1\{X_{i1}+u_i\geq0\},
\qquad
u_i\sim\operatorname{Logistic}(0,1),
\]
so that $e_0(X_i)=\Lambda(X_{i1})$. Potential outcomes satisfy
\[
Y_i(0)=\mu_0(X_i)+\varepsilon_{0i},
\qquad
Y_i(1)=\mu_0(X_i)+\theta(X_i)+\varepsilon_{1i},
\]
where $\varepsilon_{0i},\varepsilon_{1i}\overset{iid}{\sim}N(0,1)$ are independent of $(X_i,D_i)$. We test
\[
H_0:\theta(X_i)=1
\quad\text{a.s. among the treated}
\]
against
\[
\theta(X_i)=1+\delta a(X_i),
\qquad
\delta\in\{0,0.1,0.2,0.3,0.5\}.
\]
The two designs are
\[
\begin{array}{lll}
H_{31}:&
\mu_0(X)=X_1,
&
a(X)=X_2^2-1,
\\[0.15cm]
H_{32}:&
\mu_0(X)=X_1+\frac12X_2^2,
&
a(X)=X_1X_2.
\end{array}
\]

The nuisance functions $\mu_0$ and $e_0$ are estimated by five-fold cross-fitted post-Lasso using an $80$-term dictionary of polynomial and trigonometric basis functions, with the penalty and active set recomputed in each training fold. Thus, estimation does not exploit the low-dimensional functional forms of the DGP. We use the doubly robust CATT score and the finite-dimensional treated-population correction from Example~3, together with a Gaussian kernel with median-distance bandwidth.

Table~\ref{tab:design3-catt} shows accurate size in both designs and sample sizes, with null rejection frequencies between $4.1\%$ and $6.1\%$. Power increases with $\delta$ and $n$: under $H_{31}$, rejection reaches $80.0\%$ for $n=500$ and $\delta=0.5$; power is lower but still increasing under $H_{32}$. Comparisons with the non-orthogonal plug-in test are reported in the supplementary material.

\begin{table}[!htbp]
\centering
\small
\caption{Design 3: Constant CATT with high-dimensional post-Lasso.}
\label{tab:design3-catt}
\setlength{\tabcolsep}{6pt}
\begin{tabular}{|cc|ccccc|}
\hline
$n$ & Design
& $\delta=0$ & $\delta=0.1$ & $\delta=0.2$ & $\delta=0.3$ & $\delta=0.5$
\\
\hline
250 & $H_{31}$ & 4.2 & 5.6 & 10.0 & 18.9 & 44.2 \\
250 & $H_{32}$ & 6.1 & 5.2 &  6.6 & 10.5 & 17.6 \\
500 & $H_{31}$ & 4.9 & 7.8 & 17.5 & 38.9 & 80.0 \\
500 & $H_{32}$ & 4.1 & 6.4 &  9.2 & 19.4 & 42.7 \\
\hline
\end{tabular}

\smallskip
\parbox{0.9\linewidth}{\footnotesize
\emph{Notes.} $\mu_0$ and $e_0$ are estimated by five-fold cross-fitted post-Lasso using an $80$-term polynomial and trigonometric dictionary.}
\end{table}
%
%
\subsection{Empirical Application: Constant Treatment Effects in the NSW Data}

We apply the LRK test to the National Supported Work (NSW) job-training data studied by LaLonde (1986) and Dehejia and Wahba (1999). Following these studies, we combine NSW treated participants with a nonexperimental comparison group from the Panel Study of Income Dynamics (PSID). Unlike the original randomized experiment, identification in this sample relies on unconfoundedness and overlap, making it a useful application of the doubly robust CATT score in Example~3.

We test
\[
H_0:
E[Y(1)-Y(0)\mid X,D=1]
=
\theta_0.
\]
The nuisance functions
\[
\mu_0(X)=E[Y\mid X,D=0],
\qquad
e_0(X)=P(D=1\mid X)
\]
are estimated by five-fold cross-fitted Lasso and logistic Lasso using age, education, race, marital status, high-school completion, and earnings in 1974 and 1975. We use a Gaussian kernel with median-distance bandwidth and $999$ multiplier-bootstrap replications.

The sample contains $614$ individuals: $185$ NSW participants and $429$ nonexperimental PSID comparison units. The estimated CATT is $\$1{,}106$, with cross-fitted influence-function standard error $\$798$ and $95\%$ confidence interval $[-\$458,\$2{,}670]$. The LRK statistic is $26.39$, compared with a $5\%$ bootstrap critical value of $33.90$, where both quantities are multiplied by $10^{-6}$. The bootstrap $p$-value is $0.124$, so the constant-CATT null is not rejected.

The conclusion is unchanged with random-forest or hybrid nuisance learners and alternative propensity-score trimming rules. The random-forest propensity estimates exhibit weaker overlap and require trimming for approximately one-quarter of the observations, but still yield no rejection. Full robustness results, overlap diagnostics, and descriptive subgroup estimates are reported in the supplementary material. Overall, the data provide no statistically significant evidence that the program effect varies systematically with the observed pre-treatment covariates.
\section{Conclusion}

This paper develops locally robust kernel tests for semiparametric conditional moment restrictions with high- or infinite-dimensional nuisance parameters estimated by machine learning. By combining Neyman-orthogonal RKHS-indexed empirical processes with cross-fitting, the proposed tests are first-order insensitive to nuisance estimation. We establish uniform oracle equivalence under local alternatives, characterize local power and consistency, and justify a multiplier bootstrap that avoids repeated nuisance estimation.

The framework covers high-dimensional regression, significance testing, and treatment-effect homogeneity. The simulations show satisfactory finite-sample size and power, while the NSW application illustrates its use with observational data and flexible nuisance adjustment. Extensions to dependent data, generated or recursively estimated nuisance functions, and partially identified models are left for future research.
\ifblind\else
\section*{Acknowledgments}
I thank Gabriele Gazzei for capable research assistance.
\section*{Funding}
This work was supported by MICIN/AEI/10.13039/501100011033, grant CEX2021-001181-M; Comunidad de Madrid, grants EPUC3M11 (V PRICIT) and H2019/HUM-5891; and grant PID2021-127794NB-I00 (MCI/AEI/FEDER, UE).

\medskip
\noindent\textbf{Disclosure Statement.}
The author reports there are no competing interests to declare.
\fi

\medskip
\noindent\textbf{Data Availability Statement.}
The NSW data are publicly available as the \texttt{lalonde} data set
in the \texttt{MatchIt} R package. Monte Carlo data are simulated, and replication code is available from the author upon request.

\smallskip
\noindent\textbf{Generative Artificial Intelligence Use.}
OpenAI ChatGPT (GPT-5, July 2026) assisted with editing and condensing
the exposition, \LaTeX{} and code, and consistency checks of notation
and derivations, to improve clarity and implementation. The author
reviewed and verified all output and assumes full responsibility.

\newpage
\begin{appendix}

\pagenumbering{arabic}
\renewcommand{\thepage}{S-\arabic{page}}

\begin{center}
{\Large\bfseries Supplementary Material for}\\[0.5em]
{\Large\bfseries Locally Robust Kernel Specification Tests\\
for Conditional Moment Restrictions}
\end{center}

\vspace{1em}

\noindent
This supplement contains the proofs of the main results, a discussion
and verification of the assumptions in the motivating examples,
additional Monte Carlo results, and robustness analyses for the
empirical application.

\section{Proofs of the Main Results}
\pagenumbering{arabic}
\renewcommand{\thepage}{S-\arabic{page}}

\noindent
\textsc{Proof of Proposition 1.}
By Assumption 1(c) and the definition of $\alpha_{0\gamma}(\cdot,k)$,
\begin{eqnarray}
\frac{d}{d\tau}
E[
k(X_i)\varepsilon(W_i,\gamma_{F_\tau},\theta_0)
]
&=&
\frac{d}{d\tau}
E[
V_\gamma(Z_i,k)\gamma_{F_\tau}(Z_i)
]
\notag \\
&=&
-
\frac{d}{d\tau}
E[
\alpha_{0\gamma}(Z_i,k)V_\rho(Z_i)\gamma_{F_\tau}(Z_i)
]
\notag \\
&=&
-
\frac{d}{d\tau}
E[
\alpha_{0\gamma}(Z_i,k)
\rho(W_i,\gamma_{F_\tau})
]
\notag \\
&=&
\int
\alpha_{0\gamma}(z,k)\rho(w,\gamma_0)
H(dw).
\label{1g}
\end{eqnarray}
The same argument, applied componentwise with $k$ replaced by
$A_{F_0}$, gives
\begin{equation}
\frac{d}{d\tau}
E[
A_{F_0}(X_i)\varepsilon(W_i,\gamma_{F_\tau},\theta_0)
]
=
\int
\alpha_{0\gamma}(z,A_{F_0})\rho(w,\gamma_0)
H(dw).
\label{1A}
\end{equation}

Next, implicit differentiation of eq. (3) along $F_\tau$ at
$\tau=0$ and Assumption 1(f) gives
\[
0
=
\int
A_{F_0}(x)\varepsilon(w,\eta_0)H(dw)
+
\int
\alpha_{0\gamma}(z,A_{F_0})\rho(w,\gamma_0)H(dw)
+
J_0
\left.
\frac{d\theta_{F_\tau}}{d\tau}
\right|_{\tau=0}.
\]
Since $J_0$ is nonsingular by
1(f),
\begin{equation}
\left.
\frac{d\theta_{F_\tau}}{d\tau}
\right|_{\tau=0}
=
-
J_0^{-1}
\int
\left[
A_{F_0}(x)\varepsilon(w,\eta_0)
+
\alpha_{0\gamma}(z,A_{F_0})\rho(w,\gamma_0)
\right]
H(dw).
\label{der_theta}
\end{equation}

Therefore,
\begin{equation}
\frac{d}{d\tau}
E[
k(X_i)\varepsilon(W_i,\gamma_0,\theta_{F_\tau})
]
=
G_0(k)
\left.
\frac{d\theta_{F_\tau}}{d\tau}
\right|_{\tau=0}
\label{1t}
\end{equation}

Combining (\ref{1g}), (\ref{der_theta}) and (\ref{1t}) yields
\[
\frac{d}{d\tau}
E[
k(X_i)\varepsilon(W_i,\gamma_{F_\tau},\theta_{F_\tau})
]
=
\int
\phi(w,\eta_0,\alpha_0,k)H(dw),
\]
where
\[
\phi(w,\eta_0,\alpha_0,k)
=
\alpha_{0\gamma}(z,k)\rho(w,\gamma_0)
-
G_0(k)J_0^{-1}
\left[
A_{F_0}(x)\varepsilon(w,\eta_0)
+
\alpha_{0\gamma}(z,A_{F_0})\rho(w,\gamma_0)
\right].
\]
Equivalently, with
\[
\alpha_{0\theta}(x,k)
=
G_0(k)J_0^{-1}A_{F_0}(x),
\qquad
\tilde k(x)
=
k(x)-\alpha_{0\theta}(x,k),
\]
and using linearity of $k\mapsto\alpha_{0\gamma}(\cdot,k)$,
\[
\phi(w,\eta_0,\alpha_0,k)
=
-\alpha_{0\theta}(x,k)\varepsilon(w,\eta_0)
+
\alpha_{0\gamma}(z,\tilde k)\rho(w,\gamma_0).
\]
Hence
\[
\psi(w,\eta_0,\alpha_0,k)
=
\tilde k(x)\varepsilon(w,\eta_0)
+
\alpha_{0\gamma}(z,\tilde k)\rho(w,\gamma_0).
\]

By construction,
\[
E_F[
\phi(W_i,\eta_F,\alpha_F,k)
]
=
0
\]
for all $F\in\mathcal F$ in a neighborhood of $F_0$. Hence, differentiating
this identity along $F_\tau$ at $\tau=0$ gives
\[
\frac{d}{d\tau}
E[
\phi(W_i,\eta_{F_\tau},\alpha_{F_\tau},k)
]
=
-
\int
\phi(w,\eta_0,\alpha_0,k)H(dw).
\]
Finally, since
\[
\psi(w,\eta,\alpha,k)
=
k(x)\varepsilon(w,\eta)
+
\phi(w,\eta,\alpha,k),
\]
the result follows. $\square$\bigskip

\noindent
\textsc{Proof of Theorem 1.}
For $i\in I_\ell$, write
\[
\Delta\varepsilon_{i,\ell}
=
\hat\varepsilon_{i,\ell}-\varepsilon_i,
\qquad
\Delta\rho_{i,\ell}
=
\hat\rho_{i,\ell}-\rho_i,
\]
and, when $\theta_0$ is present,
\[
\Delta A_\ell
=
\hat A_\ell-A_{F_0},
\qquad
\Delta C_n
=
\hat C_n-C_0.
\]
Let $r_{0i}=r_0(W_i,\cdot)=\varepsilon(W_i,\eta_0)B_0(X_i,\cdot)+\rho(W_i,\gamma_0)\widetilde r_{\gamma,0}(Z_i,\cdot)
$. By the Riesz representation theorem,
\[
\sup_{\|k\|_K\le1}
|\hat\nu_n(k)-\nu_n^0(k)|
=
\left\|
\frac1{\sqrt n}
\sum_{\ell=1}^L\sum_{i\in I_\ell}
(\hat r_{i,\ell}-r_{0i})
\right\|_K.
\]

When $\theta_0$ is present, define
\[
d_{i,\ell}
=
\varepsilon_i\hat A_\ell(X_i)
+
\rho_i\hat\zeta_\ell(Z_i),
\qquad
\Delta d_{i,\ell}
=
\Delta\varepsilon_{i,\ell}\hat A_\ell(X_i)
+
\Delta\rho_{i,\ell}\hat\zeta_\ell(Z_i).
\]
Using the definitions of $B_0$, $\hat B_\ell$,
$\widetilde r_{\gamma,0}$, and
$\widetilde{\hat r}_{\gamma,\ell}$ gives
\[
\hat r_{i,\ell}-r_{0i}
=
u_{i,\ell}+v_{i,\ell}+p_{i,\ell}
+q_{i,\ell}+t_{i,\ell}+s_{i,\ell},
\]
where
\[
\begin{aligned}
u_{i,\ell}
&=
\Delta\varepsilon_{i,\ell}B_0(X_i,\cdot)
+
\Delta\rho_{i,\ell}
\widetilde r_{\gamma,0}(Z_i,\cdot),\\
v_{i,\ell}
&=
\rho_i\Delta_{\gamma,\ell}(Z_i,\cdot),\\
p_{i,\ell}
&=
\Delta\rho_{i,\ell}
\Delta_{\gamma,\ell}(Z_i,\cdot),\\
q_{i,\ell}
&=
-\Delta C_n'd_{i,\ell},\\
t_{i,\ell}
&=
-\Delta C_n'\Delta d_{i,\ell},\\
s_{i,\ell}
&=
-\hat\varepsilon_{i,\ell}
C_0'\Delta A_\ell(X_i).
\end{aligned}
\]
The last three terms are absent when $\theta_0$ is not present.

Conditional on $I_\ell^c$, Assumptions~4(i) and 4(iv) imply
\[
E_{F_n}[\|u_{i,\ell}\|_K^2\mid I_\ell^c]
=
o_p(1).
\]
Hence the centered sum of $u_{i,\ell}$ over $I_\ell$, normalized by
$\sqrt n$, is $o_p(1)$. Moreover, Assumption~6 gives
\[
\|E_{F_n}[u_{i,\ell}\mid I_\ell^c]\|_K
=
o_p(n^{-1/2}).
\]
Since $L$ is fixed,
\[
\frac1{\sqrt n}
\sum_{\ell=1}^L\sum_{i\in I_\ell}u_{i,\ell}
=
o_p(1).
\]

Similarly, Assumption~4(i) implies that the centered sum of
$v_{i,\ell}$ is $o_p(1)$. For every $k\in\mathcal H_K$,
Assumption~2 gives
\[
\langle\Delta_{\gamma,\ell}(\cdot),k\rangle_K
\in\Gamma
\quad\text{w.p.a.1.}
\]
Hence the nuisance-identifying equation at $F_0$ and
$F_n=(1-n^{-1/2})F_0+n^{-1/2}H$ yield
\[
E_{F_n}[v_{i,\ell}\mid I_\ell^c]
=
n^{-1/2}
E_H[
\rho_i\Delta_{\gamma,\ell}(Z_i,\cdot)
\mid I_\ell^c].
\]
By Assumption~4(i) and conditional Cauchy--Schwarz, the last expectation
is $o_p(1)$ in $\mathcal H_K$. Therefore,
\[
\frac1{\sqrt n}
\sum_{\ell=1}^L\sum_{i\in I_\ell}v_{i,\ell}
=
o_p(1).
\]

Let
\[
E_{n,\ell}[f_i]
=
|I_\ell|^{-1}\sum_{i\in I_\ell}f_i.
\]
Then
\[
\left\|
\frac1{\sqrt n}\sum_{i\in I_\ell}p_{i,\ell}
\right\|_K
\le
\frac{|I_\ell|}{\sqrt n}
E_{n,\ell}[(\Delta\rho_{i,\ell})^2]^{1/2}
E_{n,\ell}[
\|\Delta_{\gamma,\ell}(Z_i,\cdot)\|_K^2
]^{1/2}.
\]
By cross-fitting, conditional Markov's inequality, and
Assumption~5(iii), summing over the fixed number of folds gives
\[
\frac1{\sqrt n}
\sum_{\ell=1}^L\sum_{i\in I_\ell}p_{i,\ell}
=
o_p(1).
\]

We next control the terms involving $\theta_0$. Conditional on
$I_\ell^c$, the null conditional moment and nuisance-identifying
restrictions imply
\[
E[d_{i,\ell}\mid I_\ell^c]=0.
\]
Under $F_n$, its conditional mean is $O_p(n^{-1/2})$: for the first
component this follows from eq. (10) and the common
$X$-marginal, and for the second from Assumption~4(ii) and
Cauchy--Schwarz under $H$. Its conditional second moment is $O_p(1)$
by Assumption~4(ii). Thus,
\[
\frac1{\sqrt n}
\sum_{\ell=1}^L\sum_{i\in I_\ell}d_{i,\ell}
=
O_p(1).
\]
Since $\|\Delta C_n\|_{K,p}=o_p(1)$,
\[
\frac1{\sqrt n}
\sum_{\ell=1}^L\sum_{i\in I_\ell}q_{i,\ell}
=
o_p(1).
\]

Moreover,
\[
\begin{aligned}
\left\|
\frac1{\sqrt n}\sum_{i\in I_\ell}t_{i,\ell}
\right\|_K
\lesssim
\frac{|I_\ell|}{\sqrt n}
\|\Delta C_n\|_{K,p}
\Bigg\{
&E_{n,\ell}[
(\Delta\varepsilon_{i,\ell})^2
|\hat A_\ell(X_i)|^2]
\\
&+
E_{n,\ell}[
(\Delta\rho_{i,\ell})^2
|\hat\zeta_\ell(Z_i)|^2]
\Bigg\}^{1/2}.
\end{aligned}
\]
Assumption~5(ii) and conditional Markov's inequality therefore imply
\[
\frac1{\sqrt n}
\sum_{\ell=1}^L\sum_{i\in I_\ell}t_{i,\ell}
=
o_p(1).
\]

Finally, decompose
\[
s_{i,\ell}
=
-\varepsilon_iC_0'\Delta A_\ell(X_i)
-
\Delta\varepsilon_{i,\ell}
C_0'\Delta A_\ell(X_i)
\equiv
s_{1i,\ell}+s_{2i,\ell}.
\]
Assumption~4(ii) implies that the centered sum of $s_{1i,\ell}$ is
$o_p(1)$. In addition, eq. (10) gives
\[
\|E_{F_n}[s_{1i,\ell}\mid I_\ell^c]\|_K
\le
n^{-1/2}\|C_0\|_{K,p}\|a\|_{L_2(F_X)}
E[
|\Delta A_\ell(X_i)|^2
\mid I_\ell^c]^{1/2}
=
o_p(n^{-1/2}).
\]
Thus the normalized sum of $s_{1i,\ell}$ is $o_p(1)$. For the product
term,
\[
\left\|
\frac1{\sqrt n}\sum_{i\in I_\ell}s_{2i,\ell}
\right\|_K
\le
\frac{|I_\ell|}{\sqrt n}\|C_0\|_{K,p}
E_{n,\ell}[(\Delta\varepsilon_{i,\ell})^2]^{1/2}
E_{n,\ell}[|\Delta A_\ell(X_i)|^2]^{1/2},
\]
which is $o_p(1)$ by Assumption~5(i). Consequently,
\[
\frac1{\sqrt n}
\sum_{\ell=1}^L\sum_{i\in I_\ell}s_{i,\ell}
=
o_p(1).
\]

Combining the preceding bounds,
\[
\left\|
\frac1{\sqrt n}
\sum_{\ell=1}^L\sum_{i\in I_\ell}
(\hat r_{i,\ell}-r_{0i})
\right\|_K
=
o_p(1),
\]
which proves the result.
\hfill$\square$

\noindent
\textsc{Proof of Theorem 2.}
Let
\[
\bar r_n
=
E_{F_n}[r_0(W_i,\cdot)],
\qquad
\xi_{ni}
=
r_0(W_i,\cdot)-\bar r_n.
\]
Then
\[
R_n^0
=
\frac1{\sqrt n}\sum_{i=1}^n\xi_{ni}
+
\sqrt n\,\bar r_n.
\]

By the identification of the first step, Assumption 3, and the null conditional moment restriction,
\[
E[r_0(W_i,\cdot)]
=0.
\]
Therefore,
\[
\bar r_n
=
(1-n^{-1/2})E[r_0(W_i,\cdot)]
+
n^{-1/2}E_H[r_0(W_i,\cdot)]
=
n^{-1/2}\mu_H.
\]

Assumption~7 and Jensen's inequality imply
\[
\sup_n
E_{F_n}\!\left[
\|\xi_{ni}\|_K^{2+\delta}
\right]
<\infty.
\]
Hence, for every $c>0$,
\[
E_{F_n}\!\left[
\|\xi_{ni}\|_K^2
1\{\|\xi_{ni}\|_K>c\sqrt n\}
\right]
\le
\frac{
E_{F_n}[\|\xi_{ni}\|_K^{2+\delta}]
}{
c^\delta n^{\delta/2}
}
\longrightarrow0,
\]
so the Hilbert-space Lindeberg condition holds.

Let
\[
Q_0
=
E[r_0(W_i,\cdot)\otimes r_0(W_i,\cdot)],
\qquad
Q_H
=
E_H[r_0(W_i,\cdot)\otimes r_0(W_i,\cdot)].
\]
Both operators are trace class because
$\|h\otimes h\|_1=\|h\|_K^2$, where $\|\cdot\|_1$ denotes the trace
norm. The covariance operator of $\xi_{ni}$ is
\[
\Sigma_n
=
(1-n^{-1/2})Q_0
+
n^{-1/2}Q_H
-
n^{-1}\mu_H\otimes\mu_H.
\]
Since $\Sigma_0=Q_0$,
\[
\|\Sigma_n-\Sigma_0\|_1
\le
n^{-1/2}\|Q_H-Q_0\|_1
+
n^{-1}\|\mu_H\|_K^2
\longrightarrow0.
\]

Assumption~1(b) implies that $\mathcal H_K$ is separable. The triangular-array CLT for Hilbert-valued random elements (see, e.g., Kundu et al. 2000) therefore gives
\[
\frac1{\sqrt n}\sum_{i=1}^n\xi_{ni}
\rightsquigarrow
\mathbb Z
\qquad\text{in }\mathcal H_K,
\]
where $\mathbb Z$ is centered Gaussian with covariance operator
$\Sigma_0$. Since $\sqrt n\,\bar r_n=\mu_H$,
\[
R_n^0
\rightsquigarrow
\mathbb Z+\mu_H.
\]

Finally, Theorem 1 and the Riesz representations
imply
\[
\|\hat R_n-R_n^0\|_K
=
\sup_{\|k\|_K\le1}
|\hat\nu_n(k)-\nu_n^0(k)|
=
o_p(1).
\]
Thus,
\[
\hat R_n
\rightsquigarrow
\mathbb Z+\mu_H,
\qquad
LRK_n
=
\|\hat R_n\|_K^2
\rightsquigarrow
\|\mathbb Z+\mu_H\|_K^2.
\]
Under the null $H=0$, and hence $\mu_H=0$.
\hfill$\square$

\noindent
\textsc{Proof of Proposition 2.}

\[
E[r_0(W_i,\cdot)]
=
E[E[\varepsilon_i\mid X_i]\Phi_0(X_i,\cdot)]
=
0.
\]
Since $F_0$ and $H$ have the same $X$-marginal,
\[
E_H[r_0(W_i,\cdot)]
=
E_H[E_H[\varepsilon_i\mid X_i]\Phi_0(X_i,\cdot)]
=
E[a(X_i)\Phi_0(X_i,\cdot)]
=
\mu_a.
\]
Therefore,
\[
\bar r_n
=
(1-n^{-1/2})E[r_0(W_i,\cdot)]
+
n^{-1/2}E_H[r_0(W_i,\cdot)]
=
n^{-1/2}\mu_a.
\]
Since $\Pi_{0,2}^\perp$ extends $\Pi_0^\perp$, for every
$k\in\mathcal H_K$,
\[
\Delta_a(k)
=
\langle a,\Pi_{0,2}^\perp k\rangle_{L_2(F_X)}
=
\langle(\Pi_{0,2}^\perp)'a,k\rangle_{L_2(F_X)}.
\]
Hence $(\Pi_{0,2}^\perp)'a=0$ implies
$\Delta_a(k)=0$ for every $k\in\mathcal H_K$. Conversely, if
$\Delta_a(k)=0$ for every $k\in\mathcal H_K$, then
$(\Pi_{0,2}^\perp)'a$ is orthogonal to $\mathcal H_K$ in $L_2(F_X)$.
Because integral strict positivity of $K$ is equivalent to density of
$\mathcal H_K$ in $L_2(F_X)$, it follows that
$(\Pi_{0,2}^\perp)'a=0$.
\hfill$\square$

\noindent
\textsc{Proof of Theorem \ref{thm:bootstrap}.}

Let
\[
\mathcal D_n=\sigma(W_1,\ldots,W_n),
\qquad
E^*(\cdot)=E(\cdot\mid\mathcal D_n),
\qquad
P^*(\cdot)=P(\cdot\mid\mathcal D_n),
\]
and define
\[
R_n^*
=
\frac{1}{\sqrt n}
\sum_{\ell=1}^L
\sum_{i\in I_\ell}
V_i\hat r_{i,\ell}.
\]
Then
\[
LRK_n^*
=
\|R_n^*\|_K^2
=
\frac{1}{n}
\mathbf V'\hat{\mathbf R}\mathbf V.
\]

For a metric space $\mathbb D$, let $\mathrm{BL}_1(\mathbb D)$ be the class of functions $f:\mathbb D\to\mathbb R$ satisfying
\[
\sup_{x\in\mathbb D}|f(x)|\leq1,
\qquad
|f(x)-f(y)|\leq d(x,y).
\]

Define the oracle multiplier process
\[
R_n^{0*}
=
\frac1{\sqrt n}\sum_{i=1}^nV_i r_0(W_i,\cdot).
\]
We first establish its conditional weak convergence. Its conditional
covariance operator is
\[
\Sigma_n^*
=
E^*[R_n^{0*}\otimes R_n^{0*}]
=
\frac1n\sum_{i=1}^n
r_0(W_i,\cdot)\otimes r_0(W_i,\cdot).
\]
Since
\[
\|h\otimes h\|_1=\|h\|_K^2,
\]
Assumption~7 and the law of large numbers in the space of trace-class
operators imply
\[
\left\|
\Sigma_n^*
-
E_{F_n}
[r_0(W_i,\cdot)\otimes r_0(W_i,\cdot)]
\right\|_1
=o_p(1).
\]
Moreover, writing
\[
Q_0
=
E[r_0(W_i,\cdot)\otimes r_0(W_i,\cdot)],
\qquad
Q_H
=
E_H[r_0(W_i,\cdot)\otimes r_0(W_i,\cdot)],
\]
we have
\[
E_{F_n}[r_0(W_i,\cdot)\otimes r_0(W_i,\cdot)]
=
(1-n^{-1/2})Q_0+n^{-1/2}Q_H.
\]
Since $E[r_0(W_i,\cdot)]=0$, $Q_0=\Sigma_0$, and hence
\[
\|\Sigma_n^*-\Sigma_0\|_1=o_p(1).
\]

The conditional Lindeberg condition also holds. Indeed, for every
$\epsilon>0$,
\[
\begin{aligned}
&\frac1n\sum_{i=1}^n
E^*\!\left[
V_i^2\|r_0(W_i,\cdot)\|_K^2
\mathbf 1
\left\{
|V_i|\|r_0(W_i,\cdot)\|_K>\epsilon\sqrt n
\right\}
\right]
\\
&\qquad\leq
\frac{E[|V_i|^{2+\delta}]}
     {\epsilon^\delta n^{1+\delta/2}}
\sum_{i=1}^n
\|r_0(W_i,\cdot)\|_K^{2+\delta}
=o_p(1),
\end{aligned}
\]
where the last equality follows from Assumption~7. The conditional
triangular-array CLT in a separable Hilbert space
therefore gives
\[
\sup_{f\in\mathrm{BL}_1(\mathcal H_K)}
\left|
E^*[f(R_n^{0*})]-E[f(\mathbb Z)]
\right|
\overset{p}{\longrightarrow}0.
\]

It remains to replace the oracle representers by their feasible
counterparts. Using the notation and decomposition from the proof of
Theorem 1, write
\[
\hat r_{i,\ell}-r_0(W_i,\cdot)
=
u_{i,\ell}+v_{i,\ell}+p_{i,\ell}
+q_{i,\ell}+t_{i,\ell}+s_{i,\ell},
\]
and decompose
\[
s_{i,\ell}
=
s_{1i,\ell}+s_{2i,\ell},
\]
where
\[
s_{1i,\ell}
=
-\varepsilon_iC_0'\Delta A_\ell(X_i),
\qquad
s_{2i,\ell}
=
-\Delta\varepsilon_{i,\ell}
C_0'\Delta A_\ell(X_i).
\]
The bounds established in that proof imply
\[
\frac1n
\sum_{\ell=1}^L\sum_{i\in I_\ell}
\left(
\|u_{i,\ell}\|_K^2
+\|v_{i,\ell}\|_K^2
+\|q_{i,\ell}\|_K^2
+\|s_{1i,\ell}\|_K^2
\right)
=o_p(1),
\]
and
\[
\frac1{\sqrt n}
\sum_{\ell=1}^L\sum_{i\in I_\ell}
\left(
\|p_{i,\ell}\|_K
+\|t_{i,\ell}\|_K
+\|s_{2i,\ell}\|_K
\right)
=o_p(1).
\]
Terms involving the finite-dimensional parameter are omitted when it is
absent.

For the first group, independence and normalization of the multipliers
give
\[
\begin{aligned}
&E^*\left\|
\frac1{\sqrt n}
\sum_{\ell=1}^L\sum_{i\in I_\ell}
V_i
(u_{i,\ell}+v_{i,\ell}+q_{i,\ell}+s_{1i,\ell})
\right\|_K^2
=o_p(1).
\end{aligned}
\]
For the product terms, the triangle inequality gives
\[
\begin{aligned}
&E^*\left\|
\frac1{\sqrt n}
\sum_{\ell=1}^L\sum_{i\in I_\ell}
V_i(p_{i,\ell}+t_{i,\ell}+s_{2i,\ell})
\right\|_K
\\
&\qquad\leq
\frac{E|V_i|}{\sqrt n}
\sum_{\ell=1}^L\sum_{i\in I_\ell}
\left(
\|p_{i,\ell}\|_K
+\|t_{i,\ell}\|_K
+\|s_{2i,\ell}\|_K
\right)
=o_p(1).
\end{aligned}
\]
Conditional Markov's inequality therefore yields
\[
\|R_n^*-R_n^{0*}\|_K
=o_{p^*}(1)
\qquad\text{in probability}.
\]
The conditional Slutsky theorem now implies
\[
\sup_{f\in\mathrm{BL}_1(\mathcal H_K)}
\left|
E^*[f(R_n^*)]-E[f(\mathbb Z)]
\right|
\overset{p}{\longrightarrow}0.
\]

Since $h\mapsto\|h\|_K^2$ is continuous,
\[
LRK_n^*=\|R_n^*\|_K^2
\rightsquigarrow^*
\|\mathbb Z\|_K^2
\qquad\text{in probability}.
\]
Because $\Sigma_0\neq0$, the distribution of
$\|\mathbb Z\|_K^2$ is continuous. It follows that
\[
\sup_{t\in\mathbb R}
\left|
P^*(LRK_n^*\leq t)
-
P(\|\mathbb Z\|_K^2\leq t)
\right|
\overset{p}{\longrightarrow}0
\]
and
\[
c_{1-\varsigma}^*
\overset{p}{\longrightarrow}
c_{1-\varsigma}.
\]

Under the null, Theorem 2 gives
\[
LRK_n\rightsquigarrow\|\mathbb Z\|_K^2,
\]
so Slutsky's theorem and continuity at $c_{1-\varsigma}$ yield
\[
P(LRK_n>c_{1-\varsigma}^*)\longrightarrow\varsigma.
\]
Under local alternatives,
\[
LRK_n\rightsquigarrow\|\mathbb Z+\mu_a\|_K^2,
\]
and the same argument gives
\[
P(LRK_n>c_{1-\varsigma}^*)
\longrightarrow
P(\|\mathbb Z+\mu_a\|_K^2>c_{1-\varsigma}).
\]
\hfill$\square$

\noindent
\textsc{Proof of Theorem \ref{thm:consistency}.}
For $i\in I_\ell$, let
\[
\Delta\varepsilon_{i,\ell}
=\hat\varepsilon_{i,\ell}-\varepsilon_i,\quad
\Delta\rho_{i,\ell}
=\hat\rho_{i,\ell}-\rho_i,\quad
\Delta A_\ell=\hat A_\ell-A_{F_0},\quad
\Delta C_n=\hat C_n-C_0,
\]
and define
\[
\zeta_0(z)=\alpha_{0\gamma}(z,A_{F_0}),\qquad
\hat\zeta_\ell(z)=\hat\alpha_{\gamma,\ell}(z,\hat A_\ell),
\]
\[
\Delta_{\gamma,\ell}(z,\cdot)
=
\hat r_{\gamma,\ell}(z,\cdot)-r_{\gamma,0}(z,\cdot)
-C_0'(\hat\zeta_\ell(z)-\zeta_0(z)).
\]
Quantities involving $A_{F_0}$, $C_0$, and $\zeta_0$ are omitted when no
finite-dimensional parameter is present. As in the proof of
Theorem 1,
\[
\hat r_{i,\ell}-r_0(W_i,\cdot)
=u_{i,\ell}+v_{i,\ell}+p_{i,\ell}
+q_{i,\ell}+t_{i,\ell}+s_{i,\ell},
\]
where
\[
\begin{aligned}
u_{i,\ell}
&=\Delta\varepsilon_{i,\ell}B_0(X_i,\cdot)
+\Delta\rho_{i,\ell}\widetilde r_{\gamma,0}(Z_i,\cdot),\\
v_{i,\ell}
&=\rho_i\Delta_{\gamma,\ell}(Z_i,\cdot),\qquad
p_{i,\ell}
=\Delta\rho_{i,\ell}\Delta_{\gamma,\ell}(Z_i,\cdot),\\
q_{i,\ell}
&=-\Delta C_n'd_{i,\ell},\qquad
t_{i,\ell}=-\Delta C_n'\Delta d_{i,\ell},\\
s_{i,\ell}
&=-\hat\varepsilon_{i,\ell}C_0'\Delta A_\ell(X_i),
\end{aligned}
\]
with
\[
d_{i,\ell}
=\varepsilon_i\hat A_\ell(X_i)+\rho_i\hat\zeta_\ell(Z_i),
\qquad
\Delta d_{i,\ell}
=\Delta\varepsilon_{i,\ell}\hat A_\ell(X_i)
+\Delta\rho_{i,\ell}\hat\zeta_\ell(Z_i).
\]

Let $E_{n,\ell}f_i=|I_\ell|^{-1}\sum_{i\in I_\ell}f_i$. Cross-fitting,
conditional Markov's inequality, and the fixed-alternative version of
Assumption~4 yield
\[
E_{n,\ell}\|u_{i,\ell}\|_K
+E_{n,\ell}\|v_{i,\ell}\|_K=o_p(1),
\qquad
E_{n,\ell}|d_{i,\ell}|^2=O_p(1).
\]
Moreover, conditional Cauchy--Schwarz and the fixed-alternative
product-rate conditions give
\[
\begin{aligned}
E_{n,\ell}\|p_{i,\ell}\|_K
&\leq
E_{n,\ell}\!\left[(\Delta\rho_{i,\ell})^2\right]^{1/2}
E_{n,\ell}\!\left[ \|\Delta_{\gamma,\ell}(Z_i,\cdot)\|_K^2 \right]^{1/2}
=o_p(1),\\
E_{n,\ell}\|q_{i,\ell}\|_K
&\leq
\|\Delta C_n\|_{K,p}E_{n,\ell}\!\left[|d_{i,\ell}|^2\right]^{1/2}
=o_p(1),\\
E_{n,\ell}\|t_{i,\ell}\|_K
&\leq
\|\Delta C_n\|_{K,p}
\left\{
E_{n,\ell}\!\left[
(\Delta\varepsilon_{i,\ell})^2|\hat A_\ell(X_i)|^2
\right]^{1/2}
+
E_{n,\ell}\!\left[
(\Delta\rho_{i,\ell})^2|\hat\zeta_\ell(Z_i)|^2
\right]^{1/2}
\right\}
=o_p(1),\\
E_{n,\ell}\|s_{i,\ell}\|_K
&\leq
\|C_0\|_{K,p}
\left\{
E_{n,\ell}\!\left[
\varepsilon_i^2|\Delta A_\ell(X_i)|^2
\right]^{1/2}
+
E_{n,\ell}\!\left[(\Delta\varepsilon_{i,\ell})^2\right]^{1/2}
E_{n,\ell}\!\left[|\Delta A_\ell(X_i)|^2\right]^{1/2}.
\right\}
=o_p(1).
\end{aligned}
\]
Since $L$ is fixed,
\[
\left\|
\frac1n\sum_{\ell=1}^L\sum_{i\in I_\ell}
\{\hat r_{i,\ell}-r_0(W_i,\cdot)\}
\right\|_K=o_p(1).
\]

Let $\mu=E[r_0(W_i,\cdot)]$. Since
$E\|r_0(W_i,\cdot)\|_K<\infty$, the Hilbert-space law of large numbers
and
\[
\hat R_n
=\frac1{\sqrt n}\sum_{\ell=1}^L\sum_{i\in I_\ell}\hat r_{i,\ell}
\]
imply
\[
\frac{\hat R_n}{\sqrt n}
=\frac1n\sum_{\ell=1}^L\sum_{i\in I_\ell}\hat r_{i,\ell}
\overset{p}{\longrightarrow}\mu
\qquad\text{in }\mathcal H_K.
\]
Furthermore,
\[
\Delta(k)
=E[\psi(W_i,\eta_0,\alpha_0,k)]
=\langle\mu,k\rangle_K,
\]
and hence
\[
\|\mu\|_K^2
=\sup_{\|k\|_K\leq1}\Delta(k)^2>0,
\]
where positivity follows from Proposition 2 and
$(\Pi_{0,2}^\perp)'a\neq0$. Therefore,
\[
\frac{LRK_n}{n}
=\left\|\frac{\hat R_n}{\sqrt n}\right\|_K^2
\overset{p}{\longrightarrow}
\|\mu\|_K^2
=\sup_{\|k\|_K\leq1}\Delta(k)^2>0,
\]
so $LRK_n\overset{p}{\longrightarrow}\infty$.

Finally, writing $E^*$ for expectation conditional on the sample, the
multiplier normalization and the assumed empirical second-moment bound
give
\[
E^*[LRK_n^*]
=
\frac1n\sum_{\ell=1}^L\sum_{i\in I_\ell}
\|\hat r_{i,\ell}\|_K^2
=O_p(1).
\]
Conditional Markov's inequality therefore implies
\[
c_{1-\varsigma}^*
\leq\varsigma^{-1}E^*[LRK_n^*]
=O_p(1).
\]
Consequently,
\[
P\!\left(LRK_n>c_{1-\varsigma}^*\right)\longrightarrow1.
\]
\hfill$\square$
\section{Discussion of the Assumptions and Verification in the Examples}
\label{Assumptions}
\subsection{Discussion of General Assumptions}
\label{sec:general-assumptions}

This section discusses primitive sufficient conditions for Assumptions~2--7, which
apply to any nuisance function and estimator satisfying the stated
high-level conditions. Assumption 1 is not discussed here because it is a standard smoothness condition. Verification of Assumption~8 and of Assumptions~2--7
in the three motivating examples of Section~2 is given below.

\paragraph{RKHS representation (Assumption~2).}
Assumption~2 holds whenever the population adjustment $\alpha_{0\gamma}$ is
generated by a bounded linear operator
\[
T_\gamma:\mathcal H_K\to\mathcal H_Z,
\qquad
\alpha_{0\gamma}(\cdot,k)=T_\gamma k,
\]
for some RKHS $(\mathcal H_Z,K_Z)$ on $\mathcal Z$, with
$E[\rho(W_i,\gamma_0)^2K_Z(Z_i,Z_i)]<\infty$, and analogously for the
estimated adjustment $\hat\alpha_{\gamma,\ell}$ with a (possibly
data-dependent) bounded operator $\hat T_\ell$. By the reproducing
property and the definition of the adjoint,
\[
r_{\gamma,0}(z,\cdot)
=
T_\gamma'K_Z(z,\cdot)\in\mathcal H_K,
\qquad
\alpha_{0\gamma}(z,k)
=
\langle r_{\gamma,0}(z,\cdot),k\rangle_K,
\]
and similarly $\hat r_{\gamma,\ell}(z,\cdot)=\hat T_\ell'K_Z(z,\cdot)$.
Membership of the population and estimated adjustments in $\Gamma$ (part of
Assumption~3) permits use of the identifying equation
\[
E[\rho(W_i,\gamma_0)\delta(Z_i)]=0,
\qquad \delta\in\Gamma,
\]
in the proof of Theorem 1; this holds automatically
whenever the adjustment is itself constructed as a projection onto (or
regression on) $\Gamma$, as in all three motivating examples below. Beyond
this membership requirement, Assumption~3 imposes no additional restriction
on the fold structure beyond standard $L$-fold cross-fitting with $L$ fixed.

\paragraph{Mean-square consistency, product rates, and orthogonality bias (Assumptions~4--6).}
When $\theta_0$ is present, recall
\[
\zeta_0(z)=\alpha_{0\gamma}(z,A_{F_0}),
\qquad
\hat\zeta_\ell(z)=\hat\alpha_{\gamma,\ell}(z,\hat A_\ell),
\]
and define the representer error, holding the population coefficient $C_0$
fixed,
\[
\Delta_{\gamma,\ell}(z,\cdot)
=
\hat r_{\gamma,\ell}(z,\cdot)-r_{\gamma,0}(z,\cdot)
-C_0(\cdot)'\{\hat\zeta_\ell(z)-\zeta_0(z)\};
\]
without $\theta_0$, $\Delta_{\gamma,\ell}=\hat r_{\gamma,\ell}-r_{\gamma,0}$.
Assumption~4 requires mean-square consistency of the residuals and of
$\Delta_{\gamma,\ell}$, appropriate moments and consistency of $\hat A_\ell$
and $\hat\zeta_\ell$, consistency of $\hat g_n$ and $\hat J_n$, and
boundedness of the oracle representers $B_0$ and $\widetilde r_{\gamma,0}$.
These conditions involve only the objects entering the orthogonal process
and therefore accommodate generic nuisance estimators.

A convenient sufficient condition for the representer part of
Assumption~4, when the operator representation above holds with
$r_{\gamma,0}(z,\cdot)=T_0'K_Z(z,\cdot)$ and
$\hat r_{\gamma,\ell}(z,\cdot)=\hat T_\ell'K_Z(z,\cdot)$, is
\[
E_{F_n}\!\left[
\rho_i^2\|(\hat T_\ell'-T_0')K_Z(Z_i,\cdot)\|_K^2
\,\middle|\,I_\ell^c
\right]
+
E_H\!\left[
\rho_i^2\|(\hat T_\ell'-T_0')K_Z(Z_i,\cdot)\|_K^2
\,\middle|\,I_\ell^c
\right]
=o_p(1),
\]
together with mean-square consistency of the residuals; with $\theta_0$
present, the same argument applies directly to $\Delta_{\gamma,\ell}$.

When $\Delta_{\gamma,\ell}$ arises from a finite dictionary
$\mathbf b_J=(b_1,\ldots,b_J)'$, $b_j\in\mathcal H_K$, so that
\[
\Delta_{\gamma,\ell}(Z_i,\cdot)
=
\{\hat\beta_\ell(Z_i)-\beta_0(Z_i)\}'\mathbf b_J(\cdot),
\]
then
\[
\|\Delta_{\gamma,\ell}(Z_i,\cdot)\|_K^2
\leq
\lambda_{\max}(\mathbf H_J)
\|\hat\beta_\ell(Z_i)-\beta_0(Z_i)\|_2^2,
\qquad
\mathbf H_J
=
(\langle b_j,b_m\rangle_K)_{j,m=1}^J,
\]
so a sparse estimator satisfying the usual Lasso-type rate
\[
E\!\left[
\|\hat\beta_\ell(Z_i)-\beta_0(Z_i)\|_2^2
\,\middle|\,I_\ell^c
\right]
=
O_p\!\left(\frac{s\log J}{n}\right)
\]
delivers the corresponding part of Assumption~4 whenever
$\lambda_{\max}(\mathbf H_J)s\log J/n=o(1)$; this device is used repeatedly
in the examples below.

Assumption~5 controls products generated by simultaneous estimation of the
residuals, representers, $A_{F_0}$, $\zeta_0$, and $C_0$; it is satisfied,
for example, whenever each factor in a product is $o_p(n^{-1/4})$, although
asymmetric rates are allowed and are used in Example~1 below.

Assumption~6 requires $\sqrt n\,\|b_{F_n}(\hat\eta_\ell)-b_{F_n}(\eta_0)\|_K=o_p(1)$.
A convenient sufficient condition is local Lipschitz continuity of
$\eta\mapsto b_{F_n}(\eta)$ plus a quadratic remainder,
\[
\|b_{F_n}(\eta)-b_{F_n}(\eta_0)\|_K
\leq
C\left\{
n^{-1/2}\|\eta-\eta_0\|_\eta
+\|\eta-\eta_0\|_\eta^2
\right\},
\]
locally uniformly in $\eta$, together with
$\|\hat\eta_\ell-\eta_0\|_\eta=o_p(n^{-1/4})$. The first term reflects that
orthogonality is imposed at $F_0$, whereas $F_n$ is an $n^{-1/2}$
perturbation of $F_0$; the second is the usual quadratic remainder from
Neyman orthogonality.

\paragraph{Moment condition (Assumption~7).}
Since $r_0(W_i,\cdot)=\varepsilon(W_i,\eta_0)B_0(X_i,\cdot)+\rho(W_i,\gamma_0)\widetilde r_{\gamma,0}(Z_i,\cdot)$,
Assumption~4(iv) already requires $\sup_x\|B_0(x,\cdot)\|_K$ and
$\sup_z\|\widetilde r_{\gamma,0}(z,\cdot)\|_K$ to be finite. Consequently,
Assumption~7 reduces to a moment condition on the scalar residuals alone:
\[
\sup_n E_{F_n}\!\left[|\varepsilon_i|^{2+\delta}+|\rho_i|^{2+\delta}\right]<\infty
\]
for some $\delta>0$ is sufficient, since
\[
\|r_0(W_i,\cdot)\|_K
\leq
|\varepsilon_i|\sup_x\|B_0(x,\cdot)\|_K
+
|\rho_i|\sup_z\|\widetilde r_{\gamma,0}(z,\cdot)\|_K.
\]
Thus Assumption~7 imposes no restriction beyond Assumption~4(iv) and a
standard $(2+\delta)$-moment condition on the moment and generalized-error
functions, already implicit in Assumption~1(d).

\subsection{Oracle Equivalence Under Assumption 8}
\label{sec:oracle-assumption8}

Assumption~8 gives the population moment the factorized form
$r_0(w,\cdot)=\varepsilon(w,\eta_0)\Phi_0(x,\cdot)$. In all three motivating
examples the \emph{feasible} cross-fitted moment has the same structure,
since either $\rho\equiv\varepsilon$ (Examples~1--2) or
$\alpha_{0\gamma}\equiv\hat\alpha_{\gamma,\ell}\equiv0$ (Example~3); see the
Implementation section. Under this additional structure, Assumptions~3--6
admit a single, directly verifiable specialization stated purely in terms
of the residual $\varepsilon_i$ and the feature maps $\Phi_0,\hat\Phi_\ell$,
avoiding separate tracking of $\hat A_\ell,\hat\zeta_\ell,\hat g_n,\hat J_n$.
This restated version is the one verified in the examples below.

\medskip
\noindent
\textbf{Assumption 8$^\ast$ (Feasible oracle feature representation).}
For each fold $\ell$, there exists a map
$\hat\Phi_\ell:\mathcal X\to\mathcal H_K$, constructed using only $I_\ell^c$,
such that $\hat r_{i,\ell}=\hat\varepsilon_{i,\ell}\hat\Phi_\ell(X_i,\cdot),$  $i\in I_\ell.$

\medskip
Assumption~8$^\ast$ holds automatically alongside Assumption~8 whenever
$\rho\equiv\varepsilon$ or $\alpha_{0\gamma}\equiv\hat\alpha_{\gamma,\ell}\equiv0$,
since the general FSIP construction then collapses to this factorized form
at both the population and the estimated level; Assumption~2 continues to
enter, through the construction of $\Phi_0$ and $\hat\Phi_\ell$, whenever
$\rho\equiv\varepsilon$, and is vacuous when $\alpha_{0\gamma}\equiv0$.

\medskip
\noindent
\textbf{Conditions (S1)--(S3).}
For each fold $\ell$:

\smallskip
\noindent(S1)
$
E_{F_n}\!\left[(\hat\varepsilon_{i,\ell}-\varepsilon_i)^2\mid I_\ell^c\right]=o_p(1)
$,
\[
E_{F_n}\!\left[\varepsilon_i^2\|\hat\Phi_\ell(X_i,\cdot)-\Phi_0(X_i,\cdot)\|_K^2\mid I_\ell^c\right]
+
E_H\!\left[\varepsilon_i^2\|\hat\Phi_\ell(X_i,\cdot)-\Phi_0(X_i,\cdot)\|_K^2\mid I_\ell^c\right]
=o_p(1),
\]
\[
E_{F_n}\!\left[
\|\hat\Phi_\ell(X_i,\cdot)-\Phi_0(X_i,\cdot)\|_K^2
\,\middle|\,I_\ell^c
\right]
=o_p(1).
\]
and $\sup_x\|\Phi_0(x,\cdot)\|_K<\infty$.

\smallskip
\noindent(S2)
\[
E_{F_n}\!\left[(\hat\varepsilon_{i,\ell}-\varepsilon_i)^2\mid I_\ell^c\right]^{1/2}
E_{F_n}\!\left[\|\hat\Phi_\ell(X_i,\cdot)-\Phi_0(X_i,\cdot)\|_K^2\mid I_\ell^c\right]^{1/2}
=o_p(n^{-1/2}).
\]

\smallskip
\noindent(S3)
\[
\left\|
E_{F_n}\!\left[(\hat\varepsilon_{i,\ell}-\varepsilon_i)\Phi_0(X_i,\cdot)\mid I_\ell^c\right]
\right\|_K
=o_p(n^{-1/2}).
\]

\begin{theorem}[Oracle equivalence under Assumption~8]
\label{thm:oracleequiv-simplified}
Suppose Assumptions~1, 2, 3, and~7 hold, together with Assumption~8,
Assumption~8$^\ast$, and Conditions~(S1)--(S3). Then, under $\{F_n\}$,
\[
\sup_{\|k\|_K\le1}|\hat\nu_n(k)-\nu_n^0(k)|=o_p(1).
\]
\end{theorem}

\begin{proof}
Write $\varepsilon_i=\varepsilon(W_i,\eta_0)$ and, for $i\in I_\ell$,
$\hat\varepsilon_{i,\ell}=\varepsilon(W_i,\hat\eta_\ell)$. By
Assumption~8 and~8$^\ast$,
\[
\hat\nu_n(k)-\nu_n^0(k)
=
\frac1{\sqrt n}\sum_{i=1}^n
\left[
\hat\varepsilon_{i,\ell(i)}\langle\hat\Phi_{\ell(i)}(X_i,\cdot),k\rangle_K
-
\varepsilon_i\langle\Phi_0(X_i,\cdot),k\rangle_K
\right]
=
\mathrm{I}_n(k)+\mathrm{II}_n(k)+\mathrm{III}_n(k),
\]
where
\[
\mathrm{I}_n(k)
=
\frac1{\sqrt n}\sum_{i=1}^n
(\hat\varepsilon_{i,\ell(i)}-\varepsilon_i)\langle\Phi_0(X_i,\cdot),k\rangle_K,
\qquad
\mathrm{II}_n(k)
=
\frac1{\sqrt n}\sum_{i=1}^n
\varepsilon_i\langle\hat\Phi_{\ell(i)}(X_i,\cdot)-\Phi_0(X_i,\cdot),k\rangle_K,
\]
\[
\mathrm{III}_n(k)
=
\frac1{\sqrt n}\sum_{i=1}^n
(\hat\varepsilon_{i,\ell(i)}-\varepsilon_i)
\langle\hat\Phi_{\ell(i)}(X_i,\cdot)-\Phi_0(X_i,\cdot),k\rangle_K.
\]

\smallskip
\noindent\emph{Term $\mathrm{I}_n$.}
For each fold $\ell$, center on $I_\ell^c$:
\[
\mathrm{I}_n(k)
=
\frac1{\sqrt n}\sum_\ell\sum_{i\in I_\ell}
\left\{
(\hat\varepsilon_{i,\ell}-\varepsilon_i)\langle\Phi_0(X_i,\cdot),k\rangle_K
-
\langle m_\ell,k\rangle_K
\right\}
+
\frac1{\sqrt n}\sum_\ell\sum_{i\in I_\ell}\langle m_\ell,k\rangle_K,
\]
where $m_\ell=E_{F_n}[(\hat\varepsilon_{i,\ell}-\varepsilon_i)\Phi_0(X_i,\cdot)\mid I_\ell^c]$.
Since $\sup_{\|k\|_K\le1}|\langle h,k\rangle_K|=\|h\|_K$ for $h\in\mathcal H_K$,
and, conditionally on $I_\ell^c$, the summands
\[
\hat\Delta_i
=
(\hat\varepsilon_{i,\ell}-\varepsilon_i)\Phi_0(X_i,\cdot)-m_\ell,
\qquad
i\in I_\ell,
\]
are i.i.d.\ and mean zero (hence uncorrelated), the left side of the
display equals
\[
E_{F_n}\!\left[
\left\|
\frac1{\sqrt n}\sum_{i\in I_\ell}\hat\Delta_i
\right\|_K^2
\;\middle|\;I_\ell^c
\right]
=
\frac1n\sum_{i\in I_\ell}
E_{F_n}\!\left[\|\hat\Delta_i\|_K^2\mid I_\ell^c\right]
=
\frac{n_\ell}{n}
E_{F_n}\!\left[\|\hat\Delta_i\|_K^2\mid I_\ell^c\right].
\]
Since $\mathrm{Var}(Y)\le E[Y^2]$,
\[
E_{F_n}\!\left[\|\hat\Delta_i\|_K^2\mid I_\ell^c\right]
\leq
E_{F_n}\!\left[(\hat\varepsilon_{i,\ell}-\varepsilon_i)^2\|\Phi_0(X_i,\cdot)\|_K^2\mid I_\ell^c\right]
\leq
\left(\sup_x\|\Phi_0(x,\cdot)\|_K\right)^2
E_{F_n}\!\left[(\hat\varepsilon_{i,\ell}-\varepsilon_i)^2\mid I_\ell^c\right],
\]
which is $o_p(1)$ by (S1), since $n_\ell/n=O_p(1)$. Hence the first sum
is $o_p(1)$ uniformly in $k$. For the second
sum, since observations in $I_\ell$ are i.i.d.\ given $I_\ell^c$ with
common conditional mean, $\frac1{\sqrt n}\sum_{i\in I_\ell}\langle
m_\ell,k\rangle_K = \frac{n_\ell}{\sqrt n}\langle m_\ell,k\rangle_K$ with
$n_\ell=O_p(n)$, so its supremum over $\|k\|_K\le1$ is
$O_p(\sqrt n)\|m_\ell\|_K=o_p(1)$ by (S3). Hence $\sup_k|\mathrm I_n(k)|=o_p(1)$.

\smallskip
\noindent\emph{Term $\mathrm{II}_n$.}
Write $\Delta\hat\Phi_\ell(x)=\hat\Phi_\ell(x,\cdot)-\Phi_0(x,\cdot)$ and
decompose, for each fold $\ell$,
\[
\mathrm{II}_n(k)
=
\frac1{\sqrt n}\sum_\ell\sum_{i\in I_\ell}
\left\{
\varepsilon_i\langle\Delta\hat\Phi_\ell(X_i),k\rangle_K
-
\langle b_\ell,k\rangle_K
\right\}
+
\frac1{\sqrt n}\sum_\ell\sum_{i\in I_\ell}\langle b_\ell,k\rangle_K,
\]
where $b_\ell=E_{F_n}[\varepsilon_i\Delta\hat\Phi_\ell(X_i)\mid I_\ell^c]$.

\smallskip
\noindent\emph{Bias part.}
Under $\{F_n\}$, $E_{F_n}[\varepsilon_i\mid X_i]=n^{-1/2}a(X_i)$ with
$a\in L_2(F_X)$ (eq. (10)), so
\[
b_\ell
=
n^{-1/2}
E_{F_n}\!\left[a(X_i)\Delta\hat\Phi_\ell(X_i)\mid I_\ell^c\right].
\]
By Cauchy--Schwarz,
\[
\|b_\ell\|_K
\leq
n^{-1/2}\|a\|_{L_2(F_X)}
E_{F_n}\!\left[\|\Delta\hat\Phi_\ell(X_i)\|_K^2\mid I_\ell^c\right]^{1/2}
=o_p(n^{-1/2})
\]
by (S1). Summing over the $n_\ell=O_p(n)$ observations in fold $\ell$ and
dividing by $\sqrt n$, this bias contributes $o_p(1)$.

\smallskip
\noindent\emph{Centered part.}
As for $\mathrm I_n$, $\sup_{\|k\|_K\le1}|\langle h,k\rangle_K|=\|h\|_K$,
and the summands
$\varepsilon_i\Delta\hat\Phi_\ell(X_i)-b_\ell$, $i\in I_\ell$, are i.i.d.\
and mean zero given $I_\ell^c$, so
\[
E_{F_n}\!\left[
\sup_{\|k\|_K\le1}
\left|
\frac1{\sqrt n}\sum_{i\in I_\ell}
\left\{
\varepsilon_i\langle\Delta\hat\Phi_\ell(X_i),k\rangle_K
-
\langle b_\ell,k\rangle_K
\right\}
\right|^2
\;\middle|\;I_\ell^c
\right]
=
\frac{n_\ell}{n}
E_{F_n}\!\left[\|\varepsilon_i\Delta\hat\Phi_\ell(X_i)-b_\ell\|_K^2\mid I_\ell^c\right].
\]
Since $\mathrm{Var}(Y)\le E[Y^2]$,
\[
E_{F_n}\!\left[\|\varepsilon_i\Delta\hat\Phi_\ell(X_i)-b_\ell\|_K^2\mid I_\ell^c\right]
\leq
E_{F_n}\!\left[\varepsilon_i^2\|\Delta\hat\Phi_\ell(X_i)\|_K^2\mid I_\ell^c\right]
+
E_H\!\left[\varepsilon_i^2\|\Delta\hat\Phi_\ell(X_i)\|_K^2\mid I_\ell^c\right]
=o_p(1)
\]
by (S1) (the $E_H$ term accounts for the local-alternative contribution to
the conditional second moment of $\varepsilon_i$). Hence
$\sup_k|\mathrm{II}_n(k)|=o_p(1)$.

\smallskip
\noindent\emph{Term $\mathrm{III}_n$.}
By Cauchy--Schwarz within each fold,
\[
\sup_{\|k\|_K\le1}|\mathrm{III}_n(k)|
\leq
\sqrt n\,
E_{F_n}\!\left[(\hat\varepsilon_{i,\ell}-\varepsilon_i)^2\mid I_\ell^c\right]^{1/2}
E_{F_n}\!\left[\|\hat\Phi_\ell(X_i,\cdot)-\Phi_0(X_i,\cdot)\|_K^2\mid I_\ell^c\right]^{1/2}
=o_p(1)
\]
by (S2). Combining the three terms and summing over the fixed number $L$
of folds gives $\sup_{\|k\|_K\le1}|\hat\nu_n(k)-\nu_n^0(k)|=o_p(1)$.
\end{proof}

\begin{remark}[Mapping to the general assumptions]
Conditions (S1)--(S3) are the specialization of Assumptions~4--6 to the
factorized representation of Assumption~8. (S1) plays the role of
Assumption~4(i); Assumption~4(ii)--(iii), which separately track
$\hat A_\ell-A_{F_0}$, $\hat\zeta_\ell-\zeta_0$, $\hat g_n-g_0$, and
$\hat J_n-J_0$, are absorbed into the single feature-map error
$\hat\Phi_\ell-\Phi_0$: in examples where $\Phi_0$ depends on
$\theta_0$ through $B_0(x,\cdot)=K(x,\cdot)-C_0(\cdot)'A_{F_0}(x)$ (e.g.\
Example~3), $\|\hat\Phi_\ell-\Phi_0\|_K$ decomposes by the triangle
inequality into exactly these terms, but need not be bounded separately.
(S2) is Assumption~5 with the same simplification. (S3) replaces the
abstract bias condition on $b_{F_n}(\eta)$ in Assumption~6 with an
explicit linear bias in the residual estimation error, weighted only by
the fixed, known population feature map $\Phi_0$ --- directly checkable
from the estimator's own bias properties (e.g., Lasso regularization
bias, series/kernel smoothing bias) without reference to the abstract
functional $b_{F_n}$. Assumption~3's $\Gamma$-membership clause is not
needed under Assumption~8$^\ast$.
\end{remark}

\subsection{Verification for Example 1: High-Dimensional Mean and GLM}

Recall that in Example~1 there is no finite-dimensional parameter
($\alpha_{0\theta}\equiv0$, $\tilde k=k$),
$\varepsilon\equiv\rho=Y_i-\Lambda(\gamma_F(X_i))$, and $\Gamma$ is the
mean-square closure of $\bigcup_J\Gamma_J$, where
\[
\Gamma_J=\operatorname{span}\{b_1,\ldots,b_J\}\subset\Gamma
\]
is a possibly growing dictionary, $J=J_n$, and
$\mathbf b_J(x)=(b_1(x),\ldots,b_J(x))'$ is the corresponding dictionary
vector. The nuisance estimator on fold $\ell$ is
$\hat\gamma_\ell(x)=\hat\beta_\ell'\mathbf b_J(x)\in\Gamma_J$, computed by
Lasso or post-Lasso on $I_\ell^c$. Throughout, assume that the derivative
$\lambda=\Lambda'$ satisfies
\[
0<\underline\lambda
\leq
\lambda(u)
\leq
\bar\lambda
<\infty,
\]
and write
$\lambda_0(x)=\lambda(\gamma_0(x))$ and
$\hat\lambda_\ell(x)=\lambda(\hat\gamma_\ell(x))$.

\paragraph{Assumption~2.}
For the population adjustment, suppose that, for almost every $x$, the
linear functional
\[
k
\longmapsto
-\bigl(\Pi_{\Gamma,\lambda_0}k\bigr)(x)
\]
is continuous on $\mathcal H_K$. Let its Riesz representer
$r_{\gamma,0}(x,\cdot)\in\mathcal H_K$ satisfy
\[
\left\langle
r_{\gamma,0}(x,\cdot),k
\right\rangle_K
=
-\bigl(\Pi_{\Gamma,\lambda_0}k\bigr)(x),
\qquad
k\in\mathcal H_K,
\]
and assume that
\[
\sup_x
\|r_{\gamma,0}(x,\cdot)\|_K
\leq
\bar r
<
\infty.
\]
This verifies the population part of Assumption~2.

For the corresponding finite-sieve projection, let
\[
\mathbf G_J
=
E\!\left[
\lambda_0(X_i)
\mathbf b_J(X_i)\mathbf b_J(X_i)'
\right],
\]
and suppose that its eigenvalues are bounded away from zero uniformly in
$J$ and that
\[
\sup_x
\mathbf b_J(x)'\mathbf G_J^{-1}\mathbf b_J(x)
\]
is uniformly bounded. The Riesz representer of
$k\mapsto-(\Pi_{\Gamma_J,\lambda_0}k)(x)$ is then
\[
r_{\gamma,0,J}(x,\cdot)
=
-\mathbf b_J(x)'\mathbf G_J^{-1}
E\!\left[
\lambda_0(X_i)
\mathbf b_J(X_i)K(X_i,\cdot)
\right],
\]
and satisfies
\[
\left\langle
r_{\gamma,0,J}(x,\cdot),k
\right\rangle_K
=
-\bigl(\Pi_{\Gamma_J,\lambda_0}k\bigr)(x),
\qquad
k\in\mathcal H_K.
\]
The bounded-kernel and leverage conditions imply
\[
\sup_x
\|r_{\gamma,0,J}(x,\cdot)\|_K
\leq
\bar r_J,
\qquad
\sup_J\bar r_J<\infty.
\]
The analogous condition on
\[
\hat{\mathbf G}_{J,\ell}
=
\frac{1}{n_\ell^c}
\sum_{i\in I_\ell^c}
\hat\lambda_\ell(X_i)
\mathbf b_J(X_i)\mathbf b_J(X_i)'
\]
gives the estimated representer $\hat r_{\gamma,\ell}$.

\paragraph{Assumption~8 (population) and its $L_2(F_X)$ constant.}
Define the population feature map by
\[
\Phi_0(x,\cdot)
=
K(x,\cdot)+r_{\gamma,0}(x,\cdot).
\]
The reproducing property and the definition of $r_{\gamma,0}$ give
\[
\left\langle
\Phi_0(x,\cdot),k
\right\rangle_K
=
k(x)-\bigl(\Pi_{\Gamma,\lambda_0}k\bigr)(x),
\qquad
k\in\mathcal H_K.
\]
Since $K$ is bounded, with $\sup_xK(x,x)\leq\kappa$,
\[
\sup_x
\|\Phi_0(x,\cdot)\|_K
\leq
\sqrt{\kappa}+\bar r
<
\infty,
\]
and hence
\[
E\!\left[
\|\Phi_0(X_i,\cdot)\|_K^2
\right]
<
\infty.
\]

For the $L_2(F_X)$-boundedness clause, write
\[
\|k\|_{L_2(F_X,\lambda_0)}^2
=
E\!\left[
\lambda_0(X_i)k(X_i)^2
\right].
\]
The operator $\Pi_{\Gamma,\lambda_0}$ is an orthogonal projection under
the $\lambda_0$-weighted inner product and therefore is a contraction:
\[
\|\Pi_{\Gamma,\lambda_0}k\|_{L_2(F_X,\lambda_0)}
\leq
\|k\|_{L_2(F_X,\lambda_0)}.
\]
Since
$\underline\lambda\leq\lambda_0\leq\bar\lambda$, the weighted and
unweighted $L_2(F_X)$ norms are equivalent. Consequently,
\[
\|(\Pi_0^\perp)k\|_{L_2(F_X)}
=
\|k-\Pi_{\Gamma,\lambda_0}k\|_{L_2(F_X)}
\leq
\left(
1+\sqrt{\bar\lambda/\underline\lambda}
\right)
\|k\|_{L_2(F_X)}.
\]
Thus, Assumption~8 holds with
\[
C
=
1+\sqrt{\bar\lambda/\underline\lambda}.
\]

\paragraph{Assumption~8$^\ast$ (feasible).}
The analogous finite-sieve population feature map is
\[
\Phi_{0,J}(x,\cdot)
=
K(x,\cdot)+r_{\gamma,0,J}(x,\cdot),
\]
which satisfies
\[
\left\langle
\Phi_{0,J}(x,\cdot),k
\right\rangle_K
=
k(x)-\bigl(\Pi_{\Gamma_J,\lambda_0}k\bigr)(x).
\]
The same construction using $\hat\lambda_\ell$ and $\Gamma_J$ on
$I_\ell^c$ gives
\[
\hat\Phi_\ell(x,\cdot)
=
K(x,\cdot)+\hat r_{\gamma,\ell}(x,\cdot).
\]
Since $\rho\equiv\varepsilon$,
\[
\hat r_{i,\ell}
=
\hat\varepsilon_{i,\ell}
\hat\Phi_\ell(X_i,\cdot).
\]

\paragraph{Primitive rate condition (E1).}
Let
\[
\gamma_{0,J}
=
\Pi_{\Gamma_J,\lambda_0}\gamma_0
\in\Gamma_J,
\qquad
a_J
=
\|\gamma_{0,J}-\gamma_0\|_{L_2(F_X)},
\]
and suppose that $\hat\gamma_\ell$ satisfies the standard
approximate-sparse GLM prediction rate under restricted-eigenvalue and
compatibility conditions (see, e.g., B\"uhlmann and van de Geer, 2011):
\[
E_{F_n}\!\left[
\{\hat\gamma_\ell(X_i)-\gamma_0(X_i)\}^2
\,\middle|\,
I_\ell^c
\right]
=
O_p(r_n),
\qquad
r_n
=
\frac{s\log J}{n}+a_J^2,
\]
where $s$ is the sparsity of $\gamma_{0,J}$. We require
\[
r_n=o(n^{-1/2}).
\]
When $\gamma_0\in\Gamma_J$, $a_J=0$, and this condition reduces to
\[
s\log J=o(\sqrt n).
\]

To control approximation of the population feature map, define
\[
d_J^2
=
E\!\left[
\|\Phi_{0,J}(X_i,\cdot)-\Phi_0(X_i,\cdot)\|_K^2
\right]
\]
and suppose that
\[
d_J^2=O(r_n).
\]
For example, if
\[
d_J=O(J^{-\beta_K/d_x}),
\]
this condition holds whenever
\[
J^{-2\beta_K/d_x}=O(r_n).
\]
Suppose also that
\[
\left\|
\hat{\mathbf G}_{J,\ell}^{-1}
-
\mathbf G_J^{-1}
\right\|
+
\|\hat\lambda_\ell-\lambda_0\|_{L_2(F_X)}
=
O_p(\sqrt{r_n}),
\]
and that the conditional second moments of $\varepsilon_i$ given $X_i$
under $F_n$ and $H$ are uniformly bounded. Under the Gram-matrix
conditions above, these requirements imply
\[
E_{F_n}\!\left[
\|\hat\Phi_\ell(X_i,\cdot)-\Phi_{0,J}(X_i,\cdot)\|_K^2
\,\middle|\,
I_\ell^c
\right]
=
O_p(r_n),
\]
together with the analogous bounds weighted by $\varepsilon_i^2$ under
$F_n$ and $H$.

\smallskip
\noindent
\textit{Verification of (S1).}
Since $\Lambda$ is Lipschitz with constant $\bar\lambda$,
\[
(\hat\varepsilon_{i,\ell}-\varepsilon_i)^2
\leq
\bar\lambda^2
\{\hat\gamma_\ell(X_i)-\gamma_0(X_i)\}^2,
\]
so the residual part of (S1) is $O_p(r_n)=o_p(1)$. Moreover,
\[
\hat\Phi_\ell-\Phi_0
=
(\hat\Phi_\ell-\Phi_{0,J})
+
(\Phi_{0,J}-\Phi_0).
\]
The empirical-projection and sieve-approximation conditions therefore
give
\[
E_{F_n}\!\left[
\|\hat\Phi_\ell(X_i,\cdot)-\Phi_0(X_i,\cdot)\|_K^2
\,\middle|\,
I_\ell^c
\right]
=
O_p(r_n),
\]
as well as
\[
\begin{split}
&
E_{F_n}\!\left[
\varepsilon_i^2
\|\hat\Phi_\ell(X_i,\cdot)-\Phi_0(X_i,\cdot)\|_K^2
\,\middle|\,
I_\ell^c
\right]
\\
&\qquad
+
E_H\!\left[
\varepsilon_i^2
\|\hat\Phi_\ell(X_i,\cdot)-\Phi_0(X_i,\cdot)\|_K^2
\,\middle|\,
I_\ell^c
\right]
=
O_p(r_n).
\end{split}
\]
Thus, (S1) follows from $r_n=o(1)$.

\smallskip
\noindent
\textit{Verification of (S2).}
The residual-estimation error and the feature-map error are both
$O_p(\sqrt{r_n})$ in their respective mean-square norms. Their product is
therefore
\[
O_p(r_n)
=
o_p(n^{-1/2}),
\]
which verifies (S2).

\smallskip
\noindent
\textit{Verification of (S3).}
Suppose, in addition, that $\lambda=\Lambda'$ is Lipschitz. A
first-order Taylor expansion gives
\[
\hat\varepsilon_{i,\ell}-\varepsilon_i
=
-\lambda_0(X_i)
\{\hat\gamma_\ell(X_i)-\gamma_0(X_i)\}
+
R_{i,\ell},
\]
where
\[
|R_{i,\ell}|
\leq
C
\{\hat\gamma_\ell(X_i)-\gamma_0(X_i)\}^2.
\]
Because
$\sup_x\|\Phi_0(x,\cdot)\|_K<\infty$, the contribution of the
remainder satisfies
\[
\left\|
E_{F_n}\!\left[
R_{i,\ell}\Phi_0(X_i,\cdot)
\,\middle|\,
I_\ell^c
\right]
\right\|_K
=
O_p(r_n)
=
o_p(n^{-1/2}).
\]

For the leading linear term, the first-order condition of the
$\lambda_0$-weighted projection onto the full nuisance space $\Gamma$
gives, for every $g\in\Gamma$ and $k\in\mathcal H_K$,
\[
E\!\left[
\lambda_0(X_i)g(X_i)
\{k(X_i)-\Pi_{\Gamma,\lambda_0}k(X_i)\}
\right]
=
0.
\]
Since
\[
\left\langle
\Phi_0(X_i,\cdot),k
\right\rangle_K
=
k(X_i)-\Pi_{\Gamma,\lambda_0}k(X_i),
\]
it follows that
\[
E\!\left[
\lambda_0(X_i)g(X_i)
\Phi_0(X_i,\cdot)
\right]
=
0
\qquad
\text{in }\mathcal H_K
\]
for every $g\in\Gamma$. Conditional on $I_\ell^c$,
\[
\hat\gamma_\ell-\gamma_0
=
(\hat\gamma_\ell-\gamma_{0,J})
+
(\gamma_{0,J}-\gamma_0)
\in\Gamma,
\]
because
$\hat\gamma_\ell,\gamma_{0,J}\in\Gamma_J\subset\Gamma$ and
$\gamma_0\in\Gamma$. Since $F_n$ and $F_0$ have the same $X$-marginal,
\[
E_{F_n}\!\left[
\lambda_0(X_i)
\{\hat\gamma_\ell(X_i)-\gamma_0(X_i)\}
\Phi_0(X_i,\cdot)
\,\middle|\,
I_\ell^c
\right]
=
0.
\]
Thus, orthogonality with respect to the full nuisance space eliminates
both the within-sieve regularization bias and the sieve-approximation
bias from the leading term in (S3). Only the quadratic Taylor remainder
remains, and it is $o_p(n^{-1/2})$ under
$r_n=o(n^{-1/2})$. Therefore, (S3) holds.

\paragraph{Summary of rate conditions.}
In the general sieve case, sufficient conditions for this verification
are
\[
r_n
=
\frac{s\log J}{n}+a_J^2
=
o(n^{-1/2}),
\qquad
d_J^2
=
O(r_n),
\]
together with the Gram-matrix concentration and moment conditions stated
above. Here $a_J$ is the approximation error of $\gamma_0$ by
$\Gamma_J$, whereas $d_J$ is the approximation error of the finite-sieve
feature map $\Phi_{0,J}$ to the population feature map $\Phi_0$. These
conditions imply (S1) and (S2). Condition (S3) requires no additional
sieve-approximation restriction: orthogonality with respect to the full
nuisance space $\Gamma$ annihilates both
$\hat\gamma_\ell-\gamma_{0,J}$ and
$\gamma_{0,J}-\gamma_0$ in the leading linear bias. The remaining
quadratic Taylor term is controlled by $r_n=o(n^{-1/2})$.

In Design~1, the nuisance space used for each specification is the
finite dictionary itself, $\Gamma=\Gamma_J$. Therefore,
$\Phi_{0,J}=\Phi_0$ and $d_J=0$. Since the null regression also belongs
to $\Gamma_J$, $a_J=0$, and the rate requirement reduces to
\[
s\log J=o(\sqrt n),
\]
together with the stated Gram-matrix regularity conditions. This
explains why orthogonalization protects the LRK test against
within-dictionary Lasso shrinkage bias in Design~1, unlike the
non-orthogonal plug-in procedure.
\subsection{Verification for Example 2: Significance Testing}
In Example~2,
\[
\varepsilon(W_i,\gamma)
=
\rho(W_i,\gamma)
=
Y_i-\gamma(Z_i),
\qquad
X_i=(Z_i,D_i),
\]
and no finite-dimensional parameter is present. Hence
\[
B_0(X_i,\cdot)=K(X_i,\cdot),
\qquad
\widetilde r_{\gamma,0}=r_{\gamma,0},
\]
and all conditions involving $A_{F_0}$, $C_0$, $g_0$, and $J_0$ are
vacuous. Defining $\mu_0(z,\cdot)=E[K(X_i,\cdot)\mid Z_i=z]$, the
reproducing property gives $r_{\gamma,0}(z,\cdot)=-\mu_0(z,\cdot)$. A
linear cross-fitted smoother takes the form
\[
\hat\alpha_{\gamma,\ell}(z,k)
=
-\sum_{m\in I_\ell^c}\omega_{\ell m}(z)k(X_m),
\qquad
\hat r_{\gamma,\ell}(z,\cdot)
=
-\hat\mu_\ell(z,\cdot),
\qquad
\hat\mu_\ell(z,\cdot)
=
\sum_{m\in I_\ell^c}\omega_{\ell m}(z)K(X_m,\cdot),
\]
where the weights $\omega_{\ell m}(z)$ (series, kernel, or forest-based)
must not depend on $k$ given $I_\ell^c$, ensuring linearity of the
estimated adjustment. Suppose $\sup_{x\in\mathcal X}K(x,x)\leq\kappa$.
By Jensen's inequality,
\[
\sup_z\|r_{\gamma,0}(z,\cdot)\|_K
=
\sup_z\|\mu_0(z,\cdot)\|_K
\leq\sqrt{\kappa},
\qquad
\sup_x\|B_0(x,\cdot)\|_K\leq\sqrt{\kappa},
\]
verifying Assumption~2 and the boundedness in (S1). If
$\sup_z\sum_{m\in I_\ell^c}|\omega_{\ell m}(z)|=O_p(1)$, then similarly
$\sup_z\|\hat r_{\gamma,\ell}(z,\cdot)\|_K=O_p(1)$.
\paragraph{Assumption~8 (population) and its $L_2(F_X)$ constant.}
With $x=(z,d)$,
\[
\Phi_0(x,\cdot)=K(x,\cdot)-\mu_0(z,\cdot),
\qquad
\sup_x\|\Phi_0(x,\cdot)\|_K\leq2\sqrt\kappa<\infty,
\]
so $E[\|\Phi_0(X_i,\cdot)\|_K^2]<\infty$. Since
$(\Pi_0^\perp k)(x)=k(x)-E[k(X_i)\mid Z_i=z]$ is precisely
$k\mapsto k-E[k\mid Z_i]$, the residual from conditional-mean projection, hence
\[
\|\Pi_0^\perp k\|_{L_2(F_X)}
=
\|k-E[k(X_i)\mid Z_i]\|_{L_2(F_X)}
\leq
\|k\|_{L_2(F_X)},
\]
verifying Assumption~8's $L_2(F_X)$-boundedness clause with the sharp
constant $C=1$.
\paragraph{Assumption~8$^\ast$ (feasible).}
$\hat\Phi_\ell(x,\cdot)=K(x,\cdot)-\hat\mu_\ell(z,\cdot)$ and
$\hat r_{i,\ell}=\hat\varepsilon_{i,\ell}\hat\Phi_\ell(X_i,\cdot)$, since
$\rho\equiv\varepsilon$.
\paragraph{(S1)--(S2).}
Define
\[
p_{n,\ell}^2
=
E_{F_n}\!\left[
\{\hat\gamma_\ell(Z_i)-\gamma_0(Z_i)\}^2
\,\middle|\,I_\ell^c
\right],
\qquad
q_{n,\ell}^2
=
E_{F_n}\!\left[
\|\hat\Phi_\ell(X_i,\cdot)-\Phi_0(X_i,\cdot)\|_K^2
\,\middle|\,I_\ell^c
\right].
\]
Since $\hat\varepsilon_{i,\ell}-\varepsilon_i=-\{\hat\gamma_\ell(Z_i)-\gamma_0(Z_i)\}$
exactly, residual consistency in (S1) is $p_{n,\ell}=o_p(1)$. Because
$\hat\Phi_\ell(x,\cdot)-\Phi_0(x,\cdot)=-\{\hat\mu_\ell(z,\cdot)-\mu_0(z,\cdot)\}$,
\[
q_{n,\ell}^2=E_{F_n}\!\left[\|\hat\mu_\ell(Z_i,\cdot)-\mu_0(Z_i,\cdot)\|_K^2\,\middle|\,I_\ell^c\right],
\]
the cross-fitted mean-square error of $\hat\mu_\ell$ as an estimator of
the conditional kernel-mean embedding $\mu_0$. Since $\|K(x,\cdot)\|_K
\leq\sqrt\kappa$ is uniformly bounded, estimating $\mu_0$ by applying the
weights $\omega_{\ell m}(z)$ to $K(X_m,\cdot)$ is exactly the same
nonparametric regression problem as estimating $\gamma_0(z)$ by applying
the same weights to $Y_m$, now with a bounded $\mathcal H_K$-valued
response in place of a bounded scalar response. Consequently, under the
same smoothness and bandwidth/tree-depth conditions delivering
$p_{n,\ell}=o_p(n^{-1/4})$ for a local-kernel, series, or random-forest
smoother, standard nonparametric regression theory gives
$q_{n,\ell}=O_p(p_{n,\ell})$ at the same rate, so $q_{n,\ell}=o_p(n^{-1/4})$
holds under the identical primitive condition, with no separate
assumption required. Because
$F_0$, $H$, and $F_n$ share the same $X$-marginal (Section~5.1), they
share the same conditional second moment structure in $Z_i$; if
$\max\{\sup_nE_{F_n}[\varepsilon_i^2\mid Z_i],E_H[\varepsilon_i^2\mid Z_i]\}\leq C$
a.s., then
\[
E_{F_n}\!\left[\varepsilon_i^2\|\hat\Phi_\ell(X_i,\cdot)-\Phi_0(X_i,\cdot)\|_K^2\mid I_\ell^c\right]
+
E_H[\cdots]
\leq
2Cq_{n,\ell}^2,
\]
so (S1) holds whenever $p_{n,\ell}=o_p(1)$ and $q_{n,\ell}=o_p(1)$. The
only nonvacuous condition in (S2) reduces to
$p_{n,\ell}q_{n,\ell}=o_p(n^{-1/2})$, satisfied under
$p_{n,\ell}=q_{n,\ell}=o_p(n^{-1/4})$.
\paragraph{(S3).}
Since $F_n$ and $F_0$ share the same $X$-marginal,
\[
E_{F_n}[\Phi_0(X_i,\cdot)\mid Z_i]
=
E_{F_n}[K(X_i,\cdot)\mid Z_i]-\mu_0(Z_i,\cdot)
=
\mu_0(Z_i,\cdot)-\mu_0(Z_i,\cdot)
=
0
\]
exactly. As $\hat\gamma_\ell-\gamma_0$ is a function of $Z_i$ alone given
$I_\ell^c$, the tower property gives
\[
E_{F_n}\!\left[(\hat\varepsilon_{i,\ell}-\varepsilon_i)\Phi_0(X_i,\cdot)\mid I_\ell^c\right]
=
-E_{F_n}\!\left[
\{\hat\gamma_\ell(Z_i)-\gamma_0(Z_i)\}
\underbrace{E_{F_n}[\Phi_0(X_i,\cdot)\mid Z_i]}_{=0}
\,\middle|\,I_\ell^c
\right]
=
0
\]
identically, for \emph{any} nuisance estimator $\hat\gamma_\ell$,
regardless of its convergence rate --- no further condition is needed
for (S3). Equivalently, in terms of the general bias functional,
$b_{F_n}(\gamma)-b_{F_n}(\gamma_0)=E_{F_n}[\{\gamma_0(Z_i)-\gamma(Z_i)\}\Phi_0(X_i,\cdot)]=0$
for every $\gamma\in\Gamma$, since $\gamma(Z_i)-\gamma_0(Z_i)$ is
$Z_i$-measurable and $E_{F_n}[\Phi_0(X_i,\cdot)\mid Z_i]=0$. Assumption~6
(equivalently (S3)) therefore holds exactly, with no nuisance-rate
restriction at all.
\paragraph{Summary of rate conditions.}
Every condition in this verification reduces to the single standard
nonparametric rate $p_{n,\ell}=o_p(n^{-1/4})$, achievable by any of the
local-kernel, series, or random-forest estimators of
$\gamma_0(z)=E[Y_i\mid Z_i=z]$ under standard smoothness (or, for random
forests, sparsity/smoothness) conditions on $\gamma_0$. Because
$q_{n,\ell}$ governs a bounded, $\mathcal H_K$-valued version of the very
same regression problem, it obeys the same rate,
$q_{n,\ell}=O_p(p_{n,\ell})=o_p(n^{-1/4})$, with no separate condition
required. (S1) and (S2) hold under $p_{n,\ell}=q_{n,\ell}=o_p(n^{-1/4})$;
(S3) holds \emph{exactly}, for any estimator regardless of its
convergence rate, since $F_n$ and $F_0$ share the same $X$-marginal.
This is the cleanest of the three examples: a single, standard
nonparametric rate condition suffices for (S1)-(S2), and (S3) is
rate-agnostic.
\subsection{Verification for Example 3: Constant CATE/CATT}
In Example~3, the doubly robust scores are already Neyman orthogonal in
$\gamma_0$, so
\[
\alpha_{0\gamma}(\cdot,k)=\hat\alpha_{\gamma,\ell}(\cdot,k)=0,
\qquad
r_{\gamma,0}=\hat r_{\gamma,\ell}=0,
\]
and Assumption~2 is immediate; all conditions involving $\rho$ and its
representers are vacuous. Let
\[
L_i=
\begin{cases}
1,&\text{CATE},\\
D_i,&\text{CATT},
\end{cases}
\qquad
Q_i(\gamma)=
\begin{cases}
U_i(\gamma),&\text{CATE},\\
V_i(\gamma),&\text{CATT},
\end{cases}
\qquad
\varepsilon(W_i,\eta)=Q_i(\gamma)-\theta L_i.
\]
Since $\dot\varepsilon_\theta=-L_i$ and $A_{F_0}\equiv\hat A_\ell\equiv1$,
\[
C_0(\cdot)
=
\frac{E[L_iK(X_i,\cdot)]}{E[L_i]},
\qquad
B_0(X_i,\cdot)=K(X_i,\cdot)-C_0(\cdot),
\]
with sample analogue $\hat C_n(\cdot)=\sum_iL_iK(X_i,\cdot)/\sum_iL_i$,
$\hat B_\ell(X_i,\cdot)=K(X_i,\cdot)-\hat C_n(\cdot)$, which does not
depend on $\ell$ since $\dot\varepsilon_\theta$ does not depend on
$\gamma$. Suppose
\[
\sup_xK(x,x)\leq\kappa,
\qquad
\inf_nE_{F_n}[L_i]>0.
\]
Since $0\leq L_i\leq1$, $\|C_0\|_K\leq\sqrt\kappa$ and
$\sup_x\|B_0(x,\cdot)\|_K\leq2\sqrt\kappa$, verifying the boundedness
underlying (S1). The Hilbert-space law of large numbers gives
$\|\hat C_n-C_0\|_K=O_p(n^{-1/2})$ (the $O(n^{-1/2})$ gap between
expectations under $F_n$ and $F_0$ is included in this bound), hence
$\|\hat B_\ell-B_0\|_K=O_p(n^{-1/2})$, with $\hat A_\ell-A_{F_0}=0$
identically.
\paragraph{Assumption~8 (population) and its $L_2(F_X)$ constant.}
Here $\Phi_0(x,\cdot)=B_0(x,\cdot)=K(x,\cdot)-C_0(\cdot)$, and
$\sup_x\|\Phi_0(x,\cdot)\|_K\leq2\sqrt\kappa<\infty$ gives
$E[\|\Phi_0(X_i,\cdot)\|_K^2]<\infty$. For the $L_2(F_X)$-boundedness
clause, $(\Pi_0^\perp k)(x)=\tilde k(x)=k(x)-E[k(X_i)L_i]/E[L_i]$ merely
centers $k$ by an $L$-weighted mean, a real number independent of $x$.
By Cauchy--Schwarz, $|E[k(X_i)L_i]|\leq\|k\|_{L_2(F_X)}\sqrt{E[L_i]}$
(using $L_i^2=L_i$), so the centering constant satisfies
$|E[k(X_i)L_i]/E[L_i]|\leq\|k\|_{L_2(F_X)}/\sqrt{E[L_i]}$. Since
$\|1\|_{L_2(F_X)}=1$, the triangle inequality gives
\[
\|\Pi_0^\perp k\|_{L_2(F_X)}
\leq
\left(1+\frac1{\sqrt{E[L_i]}}\right)\|k\|_{L_2(F_X)},
\]
verifying Assumption~8 with explicit constant
$C=1+E[L_i]^{-1/2}$, finite under the same overlap condition
$\inf_nE_{F_n}[L_i]>0$ already used above --- no separate condition is
needed.
\paragraph{Assumption~8$^\ast$ (feasible).}
$\hat\Phi_\ell(x,\cdot)=\hat B_\ell(x,\cdot)=K(x,\cdot)-\hat C_n(\cdot)$ and
$\hat r_{i,\ell}=\hat\varepsilon_{i,\ell}\hat\Phi_\ell(X_i,\cdot)$, since
$\alpha_{0\gamma}\equiv\hat\alpha_{\gamma,\ell}\equiv0$. Crucially,
$\Delta\hat\Phi_\ell(x):=\hat\Phi_\ell(x,\cdot)-\Phi_0(x,\cdot)=-(\hat C_n-C_0)$
does not depend on $x$ or $\ell$, so $\|\Delta\hat\Phi_\ell(X_i)\|_K=\|\hat C_n-C_0\|_K=O_p(n^{-1/2})$
uniformly over $i$.
\paragraph{(S1)--(S2).}
Define
\[
p_{e,n,\ell}^2=E_{F_n}\!\left[\{\hat e_\ell(X_i)-e_0(X_i)\}^2\mid I_\ell^c\right],
\qquad
p_{d,n,\ell}^2=E_{F_n}\!\left[\{\hat\mu_{d,\ell}(X_i)-\mu_d(X_i)\}^2\mid I_\ell^c\right],\ d=0,1,
\]
and $p_{\theta,n,\ell}=|\hat\theta_\ell-\theta_0|$. Under overlap,
uniformly bounded conditional second moments of $Y_i$, and bounded or
suitably clipped nuisance estimates,
\[
E_{F_n}\!\left[(\hat U_{i,\ell}-U_i)^2\mid I_\ell^c\right]^{1/2}
\leq
C(p_{e,n,\ell}+p_{0,n,\ell}+p_{1,n,\ell}),
\qquad
E_{F_n}\!\left[(\hat V_{i,\ell}-V_i)^2\mid I_\ell^c\right]^{1/2}
\leq
C(p_{e,n,\ell}+p_{0,n,\ell}).
\]
Since $\hat\varepsilon_{i,\ell}-\varepsilon_i=(\hat Q_{i,\ell}-Q_i)-(\hat\theta_\ell-\theta_0)L_i$,
the residual part of (S1) follows from
$p_{\theta,n,\ell}+p_{e,n,\ell}+p_{0,n,\ell}+p_{1,n,\ell}=o_p(1)$ (CATE;
drop $p_{1,n,\ell}$ for CATT). The feature-map part of (S1) is immediate:
$E_{F_n}[\varepsilon_i^2\|\Delta\hat\Phi_\ell(X_i)\|_K^2\mid I_\ell^c]
=\|\hat C_n-C_0\|_K^2E_{F_n}[\varepsilon_i^2\mid I_\ell^c]=O_p(n^{-1})=o_p(1)$
given bounded second moments of $\varepsilon_i$. For (S2), the only
nonvacuous product is
\[
E_{F_n}\!\left[(\hat\varepsilon_{i,\ell}-\varepsilon_i)^2\mid I_\ell^c\right]^{1/2}
\|\hat C_n-C_0\|_K
=
o_p(1)\cdot O_p(n^{-1/2})
=
o_p(n^{-1/2}),
\]
so ordinary mean-square consistency of the doubly robust score suffices;
no separate rate condition on $\hat C_n-C_0$ is required beyond the
$\sqrt n$-consistency already established above.
\paragraph{(S3).}
Since $C_0(\cdot)=E[L_iK(X_i,\cdot)]/E[L_i]$,
\[
E[L_i\Phi_0(X_i,\cdot)]
=
E[L_iK(X_i,\cdot)]-E[L_i]C_0(\cdot)
=
0
\]
exactly. Decomposing
$\hat\varepsilon_{i,\ell}-\varepsilon_i=(\hat Q_{i,\ell}-Q_i)-(\hat\theta_\ell-\theta_0)L_i$,
the $\theta$-term contributes
\[
(\hat\theta_\ell-\theta_0)\,
E_{F_n}[L_i\Phi_0(X_i,\cdot)\mid I_\ell^c]
=
(\hat\theta_\ell-\theta_0)\,O(n^{-1/2})
=
o_p(n^{-1/2}),
\]
using that $E_{F_n}[L_i\Phi_0(X_i,\cdot)]$ differs from its (exactly
zero) value at $F_0$ by $O(n^{-1/2})$, and $p_{\theta,n,\ell}=o_p(1)$.
The $Q$-term is the standard mixed-bias term of the doubly robust score,
\[
\left\|E_{F_n}\!\left[(\hat Q_{i,\ell}-Q_i)\Phi_0(X_i,\cdot)\mid I_\ell^c\right]\right\|_K
=
O_p\!\left(p_{e,n,\ell}(p_{0,n,\ell}+p_{1,n,\ell})\right)
\quad\text{(CATE)},
\]
or $O_p(p_{e,n,\ell}p_{0,n,\ell})$ for CATT. Hence (S3) holds under
ordinary consistency together with
\[
\sqrt n\,p_{e,n,\ell}(p_{0,n,\ell}+p_{1,n,\ell})=o_p(1)
\quad\text{(CATE)},
\qquad
\sqrt n\,p_{e,n,\ell}p_{0,n,\ell}=o_p(1)
\quad\text{(CATT)},
\]
for which $o_p(n^{-1/4})$ rates on the relevant propensity-score and
outcome-regression estimators suffice.
\paragraph{Summary of rate conditions.}
(S1) and (S2) require only ordinary consistency,
$p_{\theta,n,\ell}+p_{e,n,\ell}+p_{0,n,\ell}+p_{1,n,\ell}=o_p(1)$,
together with the parametric rate $\|\hat C_n-C_0\|_K=O_p(n^{-1/2})$
established above, which holds automatically and imposes no separate
condition on the nuisance estimators. (S3) is the only condition
genuinely requiring a rate, and it is the standard double-robustness
product condition, satisfied whenever both the propensity score and the
outcome regression(s) converge at rate $o_p(n^{-1/4})$. For example, if
$e_0$ and $\mu_d$ ($d=0,1$) are estimated by series regression on a
smooth basis (power series, splines, or Fourier terms) of dimension
$J_n$, and $e_0,\mu_d$ have $s>0$ bounded derivatives on a compact
domain in $\mathbb R^{d_x}$, standard series-estimation theory gives the
mean-squared-error-optimal rate $J_n\sim n^{d_x/(2s+d_x)}$, achieving
$\|\hat e_\ell-e_0\|_{L_2(F_X)}=O_p(n^{-s/(2s+d_x)})$ and analogously for
$\hat\mu_{d,\ell}$. This is faster than the required $o_p(n^{-1/4})$
whenever $s>d_x/2$; with $d_x=2$ as in the Design~3 simulation, any
$s>1$ (e.g., twice continuously differentiable propensity and outcome
regressions, $s=2$, giving rate $n^{-1/3}$) comfortably suffices, with
no undersmoothing needed. The same conclusion holds for kernel or
random-forest estimators achieving the analogous nonparametric rate
under standard smoothness or sparsity conditions.
\section{Further results on simulations and the empirical application}
\label{app:mc}

These are the weights for the examples without cross-fitting which are used in the simulations to save space. In this case, all nuisance functions are estimated once from the full sample in the weights (though residuals are cross-fitted),
\[
(\hat{\mathbf P}\mathbf v)_i
=
\begin{cases}
\hat b_i'(\hat{\mathbf B}'\hat{\boldsymbol\Lambda}\hat{\mathbf B})^{-}
\hat{\mathbf B}'\hat{\boldsymbol\Lambda}\mathbf v, & \text{Example 1},\\[0.4em]
\displaystyle\sum_{j=1}^n\hat\omega_j(Z_i)v_j, & \text{Example 2},\\[0.6em]
\displaystyle\frac{\mathbf L'\mathbf v}{\mathbf L'\mathbf1_n}, & \text{Example 3},
\end{cases}
\]
where $\hat{\mathbf B}$, $\hat{\boldsymbol\Lambda}=\operatorname{diag}\{\lambda(\hat\gamma(X_j)):j=1,\ldots,n\}$,
and $\hat\omega_j(z)$ are now built from all $n$ observations, with $\hat\gamma$ estimated once
from the full sample. Consequently, observation $i$ itself may contribute to its own fitted
nuisance value, and $\hat P_{ii}\neq0$ in general (for Example 3, $\hat P_{ii}=L_i/(\mathbf
L'\mathbf1_n)$).

\subsection{Additional results for Design 1}
\label{app:mc-design1}

This subsection reports additional results for the high-dimensional regression experiment in Section~6.1. We first examine size in the baseline specification in which all nonzero regression coefficients are left unpenalized. We then report the non-cross-fitted LRK and low-dimensional plug-in results omitted from Table~\ref{tab:mc-design1-size}, followed by additional power results for \(J=50\) and \(J=150\).

\paragraph{Size in the baseline specification.}
In the baseline size experiment, the linear and logistic models are
\[
Y_i
=
X_{i1}+X_{i2}+\varepsilon_i,
\qquad
\varepsilon_i\sim N(0,1),
\]
and
\[
Y_i\mid X_i
\sim
\mathrm{Bernoulli}
\left[
\Lambda(0.5X_{i1}+0.5X_{i2})
\right],
\]
respectively. Because the intercept, \(X_{i1}\), and \(X_{i2}\) are unpenalized, all nonzero coefficients in this design are exempt from regularization. The experiment therefore provides a useful benchmark in which Lasso is applied to a potentially large dictionary, but does not shrink any nonzero population coefficient.

Table~\ref{tab:mc-design1-baseline-size} reports the results. The feasible LRK procedures provide satisfactory size control across models, sample sizes, and dictionary dimensions. The cross-fitted and non-cross-fitted versions are very similar and generally track the oracle procedure. The high-dimensional plug-in test is also much better behaved than in the regularization-bias design of Table~\ref{tab:mc-design1-size}. This comparison confirms that the substantial plug-in distortions in the latter design arise when nonzero coefficients are subject to penalization.

\begin{table}[!htbp]
\centering
\caption{Design 1: Empirical size in the baseline specification. Empirical rejection probabilities are reported in percent.}
\label{tab:mc-design1-baseline-size}
\small
\renewcommand{\arraystretch}{1.08}
\setlength{\tabcolsep}{6pt}
\begin{tabular}{|cc|ccccc|}
\hline
\(n\) & \(J\)
& LRK-CF
& LRK-NoCF
& Oracle
& Plug-in
& Low-dimensional
\\
\hline
\multicolumn{7}{|c|}{\textit{Panel A: Linear model}}\\
\hline
200 &  50 & 6.5 & 6.4 & 5.3 & 5.5 & 5.4 \\
200 & 100 & 6.0 & 6.0 & 4.4 & 6.7 & 5.4 \\
200 & 150 & 6.5 & 6.4 & 5.0 & 7.3 & 5.4 \\
400 &  50 & 6.2 & 6.8 & 5.8 & 3.3 & 4.0 \\
400 & 100 & 6.2 & 6.5 & 5.0 & 3.8 & 4.0 \\
400 & 150 & 6.0 & 5.8 & 4.3 & 3.8 & 4.0 \\
\hline
\multicolumn{7}{|c|}{\textit{Panel B: Logistic model}}\\
\hline
200 &  50 & 5.1 & 5.0 & 4.2 & 5.1 & 5.3 \\
200 & 100 & 4.6 & 4.8 & 4.2 & 5.5 & 5.3 \\
200 & 150 & 4.2 & 4.5 & 4.5 & 5.6 & 5.3 \\
400 &  50 & 5.2 & 5.5 & 5.6 & 5.8 & 5.4 \\
400 & 100 & 4.5 & 5.0 & 4.7 & 6.6 & 5.4 \\
400 & 150 & 5.0 & 4.7 & 4.6 & 6.4 & 5.4 \\
\hline
\end{tabular}
\vspace{1ex}
\parbox{0.96\linewidth}{\footnotesize
\emph{Notes.} LRK-CF uses five-fold cross-fitted residuals and pooled projection weights. LRK-NoCF uses full-sample residuals and derivative weights. Oracle uses the true regression residuals and derivative weights. Plug-in is the high-dimensional non-orthogonal statistic with an active-set influence-function multiplier bootstrap. Low-dimensional is the correctly specified plug-in procedure based on the intercept, \(X_1\), and \(X_2\). The nominal level is \(5\%\). Results are based on \(1{,}000\) Monte Carlo replications and \(999\) multiplier-bootstrap replications.}
\end{table}

\paragraph{Additional results under regularization bias.}
Table~\ref{tab:mc-design1-additional-size} reports the non-cross-fitted LRK results and the correctly specified low-dimensional plug-in calibration for the regularization-bias design of Table~\ref{tab:mc-design1-size}. The low-dimensional regression includes
\[
1,\quad X_{i1},\quad X_{i2},\quad
\sin(X_{i1}),\quad \sin(X_{i2})
\]
as unpenalized regressors. It is therefore correctly specified and estimated without regularization bias.

The low-dimensional plug-in procedure remains reasonably close to the nominal level. The non-cross-fitted LRK test also exhibits satisfactory size and performs similarly to LRK-CF. Thus, the main conclusions are not driven by cross-fitting. Cross-fitting is nevertheless retained in the proposed procedure because it provides theoretical protection against overfitting under more general machine-learning estimators.

\begin{table}[!htbp]
\centering
\caption{Design 1: Additional size results under regularization bias. Empirical rejection probabilities are reported in percent.}
\label{tab:mc-design1-additional-size}
\small
\renewcommand{\arraystretch}{1.08}
\setlength{\tabcolsep}{7pt}
\begin{tabular}{|cc|cc|cc|}
\hline
&&
\multicolumn{2}{c|}{Linear}
&
\multicolumn{2}{c|}{Logistic}
\\
\cline{3-6}
\(n\) & \(J\)
& LRK-NoCF & Low-dimensional
& LRK-NoCF & Low-dimensional
\\
\hline
200 &  50 & 6.5 & 5.2 & 4.3 & 6.3 \\
200 & 100 & 6.2 & 5.2 & 4.4 & 6.3 \\
200 & 150 & 6.8 & 5.2 & 4.0 & 6.3 \\
\hline
400 &  50 & 5.7 & 5.0 & 5.3 & 5.1 \\
400 & 100 & 6.0 & 5.0 & 5.4 & 5.1 \\
400 & 150 & 5.7 & 5.0 & 6.1 & 5.1 \\
\hline
\end{tabular}
\vspace{1ex}
\parbox{0.96\linewidth}{\footnotesize
\emph{Notes.} LRK-NoCF uses full-sample residuals and derivative weights. Low-dimensional denotes the correctly specified unpenalized plug-in test with an influence-function multiplier bootstrap. The null regression contains the penalized component \(0.5\{\sin(X_1)+\sin(X_2)\}\). The nominal level is \(5\%\). Results are based on \(1{,}000\) Monte Carlo replications and \(999\) multiplier-bootstrap replications.}
\end{table}

\paragraph{Additional power results.}
Table~\ref{tab:mc-design1-additional-power} reports selected rejection probabilities for \(J=50\) and \(J=150\). The results confirm the conclusions obtained from the power curves for \(J=100\) in Figure~\ref{fig:mc-example1-power}. Under the additive dictionary, the interaction
\[
a(X)=\sin(2X_1+2X_2)
\]
lies outside the nuisance space, and rejection probabilities increase with \(\delta\) and \(n\). Under the tensor-product dictionary, \(a\) belongs to the nuisance space for every dictionary dimension considered. Accordingly, rejection frequencies remain close to their null values even for \(\delta=1\).

Increasing the dimension of the additive dictionary from \(J=50\) to \(J=150\) reduces power in some configurations, particularly when \(n=200\). This reflects the cost of projecting onto a larger nuisance space in a relatively small sample. Nevertheless, the test retains substantial power in the linear model, and power increases markedly when the sample size rises to \(n=400\).

\begin{table}[!htbp]
\centering
\caption{Design 1: Additional power results for the cross-fitted LRK test. Empirical rejection probabilities are reported in percent.}
\label{tab:mc-design1-additional-power}
\small
\renewcommand{\arraystretch}{1.08}
\setlength{\tabcolsep}{6pt}
\begin{tabular}{|ccc|cc|cc|}
\hline
&&&
\multicolumn{2}{c|}{Additive}
&
\multicolumn{2}{c|}{Tensor product}
\\
\cline{4-7}
Model & \(n\) & \(J\)
& \(\delta=0.5\) & \(\delta=1\)
& \(\delta=0.5\) & \(\delta=1\)
\\
\hline
Linear
& 200 &  50 & 28.3 & 82.6 & 6.3 & 6.3 \\
Linear
& 200 & 150 & 20.1 & 57.4 & 6.6 & 6.6 \\
Linear
& 400 &  50 & 54.1 & 98.7 & 6.5 & 6.5 \\
Linear
& 400 & 150 & 57.1 & 99.5 & 6.0 & 6.0 \\
\hline
Logistic
& 200 &  50 & 7.2 & 17.3 & 4.5 & 4.9 \\
Logistic
& 200 & 150 & 6.8 & 11.4 & 4.0 & 4.9 \\
Logistic
& 400 &  50 & 9.3 & 42.0 & 6.6 & 5.3 \\
Logistic
& 400 & 150 & 8.6 & 37.6 & 4.4 & 4.7 \\
\hline
\end{tabular}
\vspace{1ex}
\parbox{0.96\linewidth}{\footnotesize
\emph{Notes.} The table reports rejection frequencies for LRK-CF. Under the additive dictionary, the interaction direction lies outside the nuisance space. Under the tensor-product dictionary, it belongs to the nuisance space and hence does not constitute a testable direction. Results are based on \(1{,}000\) Monte Carlo replications and \(999\) multiplier-bootstrap replications.}
\end{table}

Across the complete grid
\[
\delta\in\{0,0.1,\ldots,1\},
\]
the rejection frequencies of LRK-CF and LRK-NoCF are very similar. The largest difference is approximately three percentage points and occurs in a configuration with relatively high power. The differences are smaller and change sign in most other configurations. Hence, there is no evidence of a systematic power loss from cross-fitting.
\subsection{Additional results for Design 2}
\label{app:mc-design2}

Table~\ref{tab:design2_n200} reports the results for Design~2 with \(n=200\). The conclusions are similar to those obtained with \(n=400\) in Table~\ref{tab:design2_n400}. Across all kernels and nuisance learners, empirical size ranges from \(3.5\%\) to \(5.8\%\). Power increases steadily with \(\delta\), with the distance and Gaussian kernels generally outperforming the D-GM kernel. The results are also reasonably stable across the local-kernel and random-forest nuisance estimators. Since the coefficient of \(D_i\) is \(\delta/\sqrt n\), the broadly similar rejection frequencies for \(n=200\) and \(n=400\) are consistent with the local-alternative design.

\begin{table}[!htbp]
\centering
\caption{Design 2: Additional results for significance testing with flexible nuisance estimation. Empirical rejection probabilities are reported in percent.}
\label{tab:design2_n200}

\small
\renewcommand{\arraystretch}{1.10}
\setlength{\tabcolsep}{8pt}

\begin{tabular}{|c|c|c|c|c|c|c|}
\hline
Nuisance learner
& Kernel
& \(\delta=0\)
& \(\delta=1\)
& \(\delta=2\)
& \(\delta=3\)
& \(\delta=4\)
\\
\hline
Local kernel
& Distance
& 4.1 & 14.7 & 44.5 & 76.3 & 95.3
\\
Local kernel
& Gaussian
& 4.7 & 12.1 & 37.2 & 68.5 & 90.1
\\
Local kernel
& D-GM
& 3.6 & 10.4 & 29.0 & 57.8 & 82.0
\\
\hline
Random forest
& Distance
& 5.5 & 15.3 & 42.4 & 78.6 & 94.5
\\
Random forest
& Gaussian
& 5.8 & 13.1 & 34.5 & 68.8 & 88.4
\\
Random forest
& D-GM
& 3.5 & 9.5 & 23.1 & 53.8 & 78.1
\\
\hline
\end{tabular}

\vspace{1ex}

\parbox{0.95\linewidth}{\footnotesize
\emph{Notes.} The sample size is \(n=200\), and the nominal level is \(5\%\). Nuisance quantities are estimated using five-fold cross-fitting. Results are based on \(1{,}000\) Monte Carlo replications and \(999\) multiplier-bootstrap replications.}
\end{table}

%
%
\subsection{Additional results for Design 3}
\label{app:mc-design3}
This subsection reports an additional result for the constant-CATT
experiment of Section~6.3.
\paragraph{Comparison with the non-orthogonal plug-in test.}
Table~\ref{tab:mc-design3-plugin} compares the LRK test with the
corresponding non-orthogonal plug-in test, which uses the raw kernel $K$ in
place of the orthogonalized kernel induced by $\tilde k$, under the same
post-Lasso dictionary nuisance specification as Table 4.
Since the doubly robust CATT score is already orthogonal to the
infinite-dimensional nuisance functions $(\mu_0,e_0)$, this comparison
isolates the effect of orthogonalizing with respect to the
finite-dimensional parameter $\theta_0$. The plug-in test is undersized under
the null in every configuration considered and recovers only a fraction of
the power of the LRK test under the alternative. For example, at $n=500$,
$H_{31}$, $\delta=0.5$, the LRK test rejects $80.0\%$ of the time versus
$41.8\%$ for the plug-in test; at $\delta=0.3$ in the same design, the
plug-in test rejects only $6.0\%$ of the time, compared with $38.9\%$ for
LRK. Failing to orthogonalize with respect to $\theta_0$ therefore produces
both size distortion and a substantial loss of power, even under a
nonparametric first stage that already handles $(\mu_0,e_0)$ correctly.
\begin{table}[!htbp]
\centering
\caption{Design 3: Comparison of the LRK test with the non-orthogonal
plug-in test, under the post-Lasso dictionary nuisance specification of
Table 4. Empirical rejection probabilities are reported
in percent.}
\label{tab:mc-design3-plugin}
\scriptsize
\renewcommand{\arraystretch}{1.10}
\setlength{\tabcolsep}{3pt}
\begin{tabular}{|cc|ccccc|ccccc|}
\hline
&&
\multicolumn{5}{c|}{LRK}
&
\multicolumn{5}{c|}{Plug-in}
\\
\cline{3-12}
$n$ & Design
& $\delta=0$ & $\delta=0.1$ & $\delta=0.2$ & $\delta=0.3$ & $\delta=0.5$
& $\delta=0$ & $\delta=0.1$ & $\delta=0.2$ & $\delta=0.3$ & $\delta=0.5$
\\
\hline
250 & $H_{31}$ & 4.2 & 5.6 & 10.0 & 18.9 & 44.2 & 0.4 & 0.9 & 0.3 & 1.3 & 9.3 \\
250 & $H_{32}$ & 6.1 & 5.2 & 6.6 & 10.5 & 17.6 & 0.3 & 0.2 & 0.7 & 1.1 & 2.0 \\
\hline
500 & $H_{31}$ & 4.9 & 7.8 & 17.5 & 38.9 & 80.0 & 0.3 & 0.4 & 1.6 & 6.0 & 41.8 \\
500 & $H_{32}$ & 4.1 & 6.4 & 9.2 & 19.4 & 42.7 & 0.0 & 0.4 & 0.9 & 1.7 & 11.2 \\
\hline
\end{tabular}
\vspace{1ex}
\parbox{0.96\linewidth}{\footnotesize
\emph{Notes.} LRK is the proposed test based on the orthogonalized kernel
$\tilde k$. Plug-in uses the raw kernel $K$ without the finite-dimensional
correction for $\theta_0$. Both tests use the same cross-fitted post-Lasso
dictionary estimates of $(\mu_0,e_0)$. The nominal level is $5\%$. Results
are based on $1{,}000$ Monte Carlo replications and $999$
multiplier-bootstrap replications.}
\end{table}
%
\subsection{Additional results for the empirical application}
\label{app:mc-empirical}
This subsection reports robustness results for the NSW constant-CATT
application of Section~6.4: alternative nuisance learners and
propensity-score trimming rules, overlap diagnostics, and descriptive
subgroup CATT estimates.
\paragraph{Robustness across nuisance learners and trimming rules.}
Table~\ref{tab:empirical-catt-robustness} reports the LRK test across five
specifications: the Lasso baseline reported in Section~6.4; a
random-forest specification, in which both $\mu_0$ and $e_0$ are estimated
by cross-fitted random forests; a hybrid specification, which pairs the
Lasso propensity score with a random-forest outcome regression; and two
sensitivity checks that vary the propensity-score trimming threshold. In
every specification the null hypothesis of a constant CATT is not rejected
at the $5\%$ level, and the point estimate of the CATT remains economically
similar under the two alternative trimming rules. The random-forest
specification yields a much smaller and less precisely estimated CATT
($\$201$, $p=0.696$); as discussed below, this specification also exhibits
substantially weaker propensity-score overlap, so it is best interpreted
with caution rather than as evidence against the Lasso-based baseline.
\begin{table}[!htbp]
\centering
\caption{NSW application: robustness of the constant CATT test across
nuisance learners and propensity-score trimming rules.}
\label{tab:empirical-catt-robustness}
\small
\renewcommand{\arraystretch}{1.12}
\setlength{\tabcolsep}{3.5pt}
\begin{tabular}{|l|c|c|c|c|c|c|}
\hline
Specification & $\widehat\theta_{\mathrm{CATT}}$ & 95\% CI & Stat.\ $\times10^{-6}$
& Crit.\ $\times10^{-6}$ & $p$-value & ESS \\
\hline
Baseline (Lasso) & 1,106 & [$-$458, 2,670] & 26.39 & 33.90 & 0.124 & 99.0 \\
Random forest & 201 & [$-$1,921, 2,322] & 21.10 & 65.97 & 0.696 & 59.4 \\
Hybrid & 581 & [$-$1,000, 2,161] & 23.95 & 36.57 & 0.197 & 99.0 \\
Trim $[0.01,0.99]$ & 1,108 & [$-$457, 2,672] & 26.39 & 33.90 & 0.124 & 99.0 \\
Trim $[0.05,0.95]$ & 1,098 & [$-$466, 2,662] & 26.44 & 33.89 & 0.124 & 99.6 \\
\hline
\end{tabular}
\vspace{1ex}
\parbox{0.96\linewidth}{\footnotesize
\emph{Notes.} Hybrid uses the Lasso propensity score with a random-forest
outcome regression. Trim $[a,1-a]$ specifications use the Lasso nuisance fits
with propensity scores truncated to $[a,1-a]$; the baseline uses $[0.03,0.97]$.
Confidence intervals for $\widehat\theta_{\mathrm{CATT}}$ use the cross-fitted doubly
robust score standard error. All specifications use the same sample, five
cross-fitting folds, Gaussian kernel with median-distance bandwidth, and 999
multiplier-bootstrap replications. The statistic and critical value are
multiplied by $10^{-6}$.}
\end{table}
\paragraph{Overlap diagnostics.}
Table~\ref{tab:empirical-catt-overlap} reports the distribution of the
estimated propensity score and the extent of trimming under each
specification. Under the Lasso baseline, overlap is reasonable: only
$1.1\%$ of observations are affected by the $[0.03,0.97]$ trimming rule, and
the effective sample size of control units, weighted by the odds
$\hat e(X)/\{1-\hat e(X)\}$, is $99.0$ out of $429$ controls. Overlap is
markedly worse under the random-forest propensity score, which produces
estimated propensities as low as $0.0002$ and as high as $0.93$, trims
$24.3\%$ of the sample, and reduces the control-weight effective sample size
to $59.4$. This is consistent with the well-documented limited overlap
between NSW participants and the PSID comparison group, and it is the
reason we treat the random-forest result in
Table~\ref{tab:empirical-catt-robustness} as less reliable than the Lasso
baseline.
\begin{table}[!htbp]
\centering
\caption{NSW application: propensity-score overlap diagnostics.}
\label{tab:empirical-catt-overlap}
\small
\renewcommand{\arraystretch}{1.12}
\setlength{\tabcolsep}{5pt}
\begin{tabular}{|l|c|c|c|c|c|c|}
\hline
Specification & Min & 5th pct.\ & 95th pct.\ & Max & Trimmed $n$ & Trimmed \% \\
\hline
Baseline (Lasso) & 0.025 & 0.045 & 0.767 & 0.969 & 7 & 1.1 \\
Random forest & 0.000 & 0.007 & 0.838 & 0.926 & 149 & 24.3 \\
Hybrid & 0.025 & 0.045 & 0.767 & 0.969 & 7 & 1.1 \\
Trim $[0.01,0.99]$ & 0.025 & 0.045 & 0.767 & 0.969 & 0 & 0.0 \\
Trim $[0.05,0.95]$ & 0.025 & 0.045 & 0.767 & 0.969 & 45 & 7.3 \\
\hline
\end{tabular}
\vspace{1ex}
\parbox{0.96\linewidth}{\footnotesize
\emph{Notes.} Quantiles are computed from the untrimmed cross-fitted
propensity scores. Hybrid and the two trimming sensitivity checks share the
Lasso propensity score with the baseline and differ only in the truncation
rule or the outcome-regression learner.}
\end{table}

\paragraph{Descriptive subgroup estimates.}
Table~\ref{tab:empirical-catt-subgroup} reports descriptive CATT estimates
for eight covariate subgroups, obtained from the cross-fitted doubly robust
pseudo-outcome underlying the baseline test. These estimates are intended
only to aid interpretation and do not constitute formal tests of
treatment-effect heterogeneity. Seven of the eight confidence intervals
include zero. The exception is the subgroup at or above the treated
median age, with an estimated CATT of $\$2,691$ and a $95\%$ confidence
interval of $[\$351,\$5,030]$. This interval tests whether the effect in
that subgroup is zero, not whether it differs from the effect in another
subgroup or from the overall CATT. It therefore does not contradict the
LRK test's failure to reject the constant-CATT null.

\begin{table}[!htbp]
\centering
\caption{NSW application: descriptive subgroup CATT estimates. These
intervals are purely descriptive and do not correct for testing multiple
subgroups; see the discussion in the text.}
\label{tab:empirical-catt-subgroup}
\small
\renewcommand{\arraystretch}{1.12}
\setlength{\tabcolsep}{6pt}
\begin{tabular}{|l|c|c|c|c|}
\hline
Subgroup & Treated $n$ & CATT & SE & 95\% descriptive CI \\
\hline
Age $<$ treated median & 88 & $-$641 & 1,030 & [$-$2,659, 1,378] \\
Age $\geq$ treated median & 97 & 2,691 & 1,194 & [351, 5,030] \\
Education $<12$ & 131 & 1,109 & 922 & [$-$699, 2,916] \\
Education $\geq12$ & 54 & 1,099 & 1,572 & [$-$1,982, 4,181] \\
Zero 1975 earnings & 111 & 1,553 & 1,035 & [$-$474, 3,581] \\
Positive 1975 earnings & 74 & 436 & 1,257 & [$-$2,028, 2,899] \\
Black & 156 & 1,043 & 913 & [$-$747, 2,833] \\
Non-Black & 29 & 1,447 & 1,326 & [$-$1,151, 4,046] \\
\hline
\end{tabular}
\vspace{1ex}
\parbox{0.96\linewidth}{\footnotesize
\emph{Notes.} CATT estimates use the baseline (Lasso) cross-fitted doubly
robust pseudo-outcome, restricted to each subgroup; standard errors use the
corresponding subgroup-restricted influence function. Confidence intervals are descriptive, use a pointwise (not
simultaneous) $95\%$ level, and are not adjusted for testing eight
subgroups.}
\end{table}

\end{appendix}

\end{document}